\documentclass[10pt,twocolumn,twoside]{IEEEtran}
\usepackage{amsmath}
\usepackage{amsfonts}
\usepackage{amssymb}
\usepackage{framed}
\usepackage{diagbox}
\usepackage{hyperref}
\hypersetup{hypertex=true,
            colorlinks=true,
            linkcolor=blue,
            anchorcolor=blue,
            citecolor=blue}
\usepackage{latexsym}
\usepackage{graphicx}
\usepackage[ruled,linesnumbered]{algorithm2e}
\SetKwRepeat{Do}{do}{while}%
\usepackage{epsfig,psfrag,color}
\usepackage{epstopdf}
\usepackage[justification=centering]{caption}

\newtheorem{theorem}{\bf Theorem}[section]
\newtheorem{lemma}[theorem]{\bf Lemma}
\newtheorem{proposition}[theorem]{\bf Proposition}

\newtheorem{note}[theorem]{\bf Note}

\newtheorem{remark}[theorem]{\bf Remark}

\newtheorem{definition}[theorem]{\bf Definition}
\newenvironment{proof}{\noindent{\em Proof:}}{\quad \hfill$\Box$\vspace{2ex}}

\begin{document}

\title{Random phaseless  sampling    for   causal    signals   in  shift-invariant spaces: a zero distribution  perspective}

\author{Youfa~Li, Wenchang~Sun
\thanks{Youfa~Li. College of Mathematics and Information Science,
Guangxi University,  Naning  530004, China.   Email: youfalee@hotmail.com}
\thanks{Wenchang~Sun. School   of Mathematical Sciences and LPMC,
Nankai University, Tianjin 300071, China. Email:sunwch@nankai.edu.cn}

\thanks{Youfa Li is partially supported by Natural Science Foundation of China (Nos: 61961003,  61561006, 11501132), Natural Science Foundation of Guangxi (No: 2019GXNSFAA185035)
and the talent project of  Education Department of Guangxi Government  for Young-Middle-Aged backbone teachers. Wenchang Sun
was partially supported by the
National Natural Science Foundation of China (11525104, 11531013 and 11761131002).}}

\date{}

\maketitle

\begin{abstract}
We      proved    that   the phaseless sampling (PLS)
in  the  linear-phase modulated  shift-invariant space (SIS) $V(e^{\textbf{i}\alpha \cdot}\varphi), \alpha\neq0,$
is impossible even though the real-valued  function  $\varphi$ enjoys the full spark property (so does $e^{\textbf{i}\alpha \cdot}\varphi$).
Stated another way, the PLS in the  complex-generated
SISs is essentially different from that in the real-generated  ones.
Motivated by this,  we first  establish  the condition on the complex-valued generator   $\phi$ such that the PLS of  nonseparable causal  (NC) signals  in  $V(\phi)$  can be achieved by random sampling. The condition is established  from the     generalized Haar condition (GHC) perspective. Based on the proposed reconstruction approach,  it is  proved that if the GHC holds then with probability $1$,
the random     sampling density (SD) $=3$ is sufficient for the PLS of   NC  signals   in the  complex-generated  SISs.
For the real-valued  case  we also   prove   that,   if the GHC holds then with probability $1$,
the random  SD $=2$ is sufficient for the PLS of  real-valued NC signals   in the  real-generated  SISs.
For the local reconstruction of  highly oscillatory  signals such as chirps,
a great number of  deterministic  samples are   required.
Compared with deterministic sampling, the  proposed  random approach
enjoys  not only the greater   sampling flexibility but the    much  smaller number of samples.
 To verify our results, numerical simulations were  conducted  to  reconstruct highly oscillatory NC signals   in the  chirp-modulated SISs.
\end{abstract}

\bigskip
\begin{IEEEkeywords} Random  phaseless sampling, complex (real)-generated shift-invariant space, generalized Haar condition, sampling density, highly oscillatory  signals.\end{IEEEkeywords}
\bigskip


\section{Introduction}\label{section1}
Phase retrieval (PR) is a nonlinear   problem  that seeks to
reconstruct   a signal   $f$, \emph{up to  a unimodular scalar},  from   the intensities  of  the   linear  measurements
(c.f.\cite{Gerchberg,Fienup1,Fienup2,Fienup3,Cand,Shenoy12,reasonforPR})
\begin{align}\notag \begin{array}{lll}  b_{k}:=|\langle f, \textbf{a}_{k}\rangle|, \ k\in \Gamma, \end{array}\end{align}
where $\textbf{a}_{k}$ is called  the \emph{measurement vector}.

As stated in   Y. Shechtman et. al \cite{reasonforPR}
one of the     reasons for PR in optics is that,  for  highly  oscillatory  signals such as optical waves (electromagnetic fields oscillating at $10^{15}$ Hz and higher), measuring their phases is very difficult or even  impossible for electronic measurement
devices. PR   has   been
widely investigated  in  engineering and mathematical  problems such as  coherent diffraction imaging (\cite{reasonforPR,crystallography,Candes0}),   quantum tomography (\cite{Heinosaarri}),
and    frame theory (\cite{HanDe1,HanDe3}).
A concrete  PR problem     corresponds  to
the specific  signal class  $\mathcal{C}$   and     measurement vectors (e.g. \cite{example,Fienup46,Fienup278,Fienup288,Science61,Eldar,FAA}).
For example, Alaifari et. al \cite{Fienup278} considered the PR of  real-valued bandlimited functions by frame  measurement vectors.
When
$f$ lies in a function class $\mathcal{C}$    and    $\textbf{a}_{k}$   is  the  shift  of the  Dirac distribution,
then the corresponding PR is the \emph{phaseless sampling} (PLS for short), modeled as
\begin{align}\notag   \hbox{to reconstruct} \ f \ \hbox{by the  samples} \ |f(x)|, x\in \Omega, \end{align}
up to  a unimodular scalar.
In what follows, we introduce  the recent developments on PLS in shift-invariant spaces (SISs).
\subsection{Related work}
SIS has   many applications in   signal processing.
Please refer to   \cite{FAA,Fienup289,SHI,Zayed0} and the references therein for a few examples. For a  generator $g: \mathbb{R}\rightarrow \mathbb{C}$,
its   SIS is defined as \begin{align}\notag \begin{array}{lll}  V(g):=\{\sum_{k\in \mathbb{Z}}c_{k}g(\cdot-k): \{c_{k}\}_{k\in \mathbb{Z}}\in \ell^{2}\}, \end{array}\end{align}
where $ \{c_{k}\}_{k\in \mathbb{Z}}\in \ell^{2}$ means $\sum_{k\in \mathbb{Z}} |c_{k}|^{2}<\infty.$
Recently, PLS in SISs   received  much attention  (e.g.\cite{Sun2,Sun1,Iserles,QiyuPR,QiyuPR1,Jaming}).
Particularly, it was investigated   for  bandlimited signals in Thakur \cite{Iserles}, P. Jaming, K. Kellay and R. Perez Iii \cite{Jaming}
and C.K. Lai, F. Littmann, E. Weber \cite{Fienup45}.
Note that the spaces  of bandlimited signals  are shift-invariant and the corresponding generators (sinc function or  its dilations)  are
infinitely supported (c.f. \cite{BINHAN,BHBOOK}).
 Chen,  Cheng,  Sun and   Wang  \cite{QiyuPR} established the
PLS  of   nonseparable (the definition of nonseparability is postponed to section \ref{ppp786}) real-valued  signals  in  the    SIS
 from a  compactly supported generator.
 W. Sun \cite{Sun1} established  the
PLS  for nonseparable real-valued  signals  in SISs generated by B-splines.

 Note  that the   generators  and signals in \cite{Sun1,QiyuPR}
 are all real-valued, and the sampling is deterministic.
Motivated by the  results therein we will  investigate the random PLS of causal signals  in complex (or real)-generated SISs.
Here  a signal $f\in V(g)$ is  said to be   causal if
\begin{align}\notag \begin{array}{lll}  f=\sum_{k=0}^{\infty}c_kg(\cdot-k), \ c_{0}\neq0.\end{array}\end{align}
 The set of causal signals in $V(g)$ is denoted by  $V_{\hbox{ca}}(g)$. Causal signals are an important class  of signals (c.f. \cite{FAA,causal1,causal2}). Particularly, the Fourier measurement-based  PR of causal signals in SISs was  addressed in \cite{FAA}.
 In what follows,
we introduce the  motivation.

\subsection{Motivation}
\subsubsection{Full spark property fails for complex-valued case}\label{ppp786}
Many practical applications require    processing    signals    in the
SISs    from   complex-valued generators such as chirps (e.g.\cite{SHI,chirp}).
We will investigate the PLS in complex-generated  SISs.
To the best of our knowledge, there are few  literatures on this topic.
We are greatly   motivated by   Theorem
\ref{examm},  which will state that the PLS  in  the complex-generated   SISs is essentially different from that in  the real-generated ones.

Some denotations and definitions  are  necessary    for Theorem
\ref{examm}. A nonzero function $f$ is traditionally denoted as $f\not\equiv0$,
and $f(x)\neq0$ means that the point $x$ is not the zero  of $f$.
The conjugate of $a\in \mathbb{C}$ is denoted by $\bar{a}.$ The real and imaginary parts of $a$
are denoted by $\Re(a)$ and $\Im(a)$, respectively.
Any   $a\neq0$  can be  denoted by   $|a|e^{\textbf{i}\theta(a)}$ where    $\textbf{i}$, $|a|$ and  $\theta(a)$ are   the imaginary unit,   modulus and     phase, respectively. For   phases $\theta(a)$ and $ \theta(b)$,  we say that  $\theta(a)=\theta(b)$ if  $\theta(a)=\theta(b)+2k\pi$ for a certain  $k\in \mathbb{Z}.$
Traditionally,   the  phase of zero can be   assigned   arbitrarily.

Throughout this paper the complex and real-valued generators are denoted by
$\phi$ and $\varphi$, respectively. Without loss of generality, assume that  \begin{align}\label{support}  \hbox{supp}(\phi)\subseteq(0,s), \ \hbox{supp}(\varphi)\subseteq(0,s)  \end{align} with the  integer  $s\geq2$.
A function $0\not\equiv f\in V(\phi)$ (or $V(\varphi)$) is  separable if there exist    $0\not\equiv f_{1}$ and $ 0\not\equiv f_{2}\in V(\phi)$ (or $V(\varphi)$) such that
$f=f_{1}+f_{2}$ and  $f_{1}f_{2}\equiv0$. Clearly, if   $f$ is separable then $|f|=|f_{1}+e^{\textbf{i}\alpha}f_{2}|$ where $\alpha\in (0, 2\pi)$, and consequently  it  is not distinguishable from $f_{1}+e^{\textbf{i}\alpha}f_{2}$ by  the samples of $|f|$.

%

 For the  above    real-valued generator  $\varphi$,  if  the    matrix \begin{align}\label{fullspark}\begin{array}{lllllllllllllllll}  \big(\varphi(x_{k}+n)\big)_{1\leq k\leq 2s-1, 0\leq n\leq s-1}\end{array}\end{align}  is   full spark (c.f. \cite{HanDe4,HanDe5}) for  any   $2s-1$ distinct  points  $x_{k}\in(0,1), k=1, \ldots, 2s-1$, namely,
every $s\times s$ submatrix is  nonsingular,
then it follows from    \cite{QiyuPR} that the  real-valued   nonseparable signals  in $V(\varphi)$
can be determined  by sufficiently many  samples.
The  B-spline generators in \cite{Sun1}   satisfy the
property. However,  the following theorem  implies that the  property is not sufficient for achieving PLS  when the generator is  complex-valued.

\begin{theorem}\label{examm}
Let  $\varphi$ be  real-valued    such that  $\hbox{supp}(\varphi)\subseteq(0,s)$ and  the   matrix     in  \eqref{fullspark} is full spark for   any $2s-1$  distinct   points  $x_{k}\in(0,1)$, $k=1, \ldots, 2s-1$.
Define  $\phi:=e^{\textbf{i}\alpha \cdot}\varphi$ with $\alpha\neq0$.
Then the PLS  in   $V_{\hbox{ca}}(\phi)$ can not be achieved despite the fact that  $\phi$ also satisfies the full spark property.
\end{theorem}
\begin{proof}
It is easy to check that   $\phi$  inherits the full spark property of $\varphi$.
%
It follows from the full spark property that     $\{\phi(\cdot+k): k=0, \ldots,s-1\}$
is  linearly independent. We first choose   $\beta\in \mathbb{R}$ such that $\alpha-\beta\neq k\pi$
for any $k\in \mathbb{Z}.$
Let  $N\geq2$.
Define a    sequence $\{c_{k}\}^{N}_{k=0}$    such that $c_{0}=1$ and $c_{1}=e^{\textbf{i}\beta}$.
It is easy to check that
$\{c_{k}\}^{N}_{k=0}\neq e^{\textbf{i}\widehat{\theta}}
\{e^{\textbf{i}2\alpha k}\overline{c}_{k}\}^{N}_{k=0} $ for any $\widehat{\theta}\in [0, 2\pi)$.
By the above linear independence, we have
$\sum^{N}_{k=0} c_{k}\phi(\cdot-k)\neq e^{\textbf{i}\widehat{\theta}}\sum^{N}_{k=0} e^{\textbf{i}2\alpha k}\overline{c}_{k}\phi(\cdot-k)$. However, it is easy to check that
 \begin{align}\label{ambi} |\sum^{N}_{k=0} e^{\textbf{i}2\alpha k}\overline{c}_{k}\phi(\cdot-k)|=|\sum^{N}_{k=0} c_{k}\phi(\cdot-k)|.\end{align}
In other words, the PLS in
$V_{\hbox{ca}}(\phi)$ can not be achieved.
\end{proof}

Motivated by Theorem \ref{examm}, we need  to establish a condition on the complex-valued generator  $\phi$
such that the  PLS  in
$V_{\hbox{ca}}(\phi)$  can  be achieved.
The  condition will be established  from the zero distribution (or the Lebesgue measure of zero set) perspective. Our motivation for this perspective is introduced   in what follows.

\subsubsection{New perspective:  zero distribution-based PLS}\label{newperspective}
We  first    interpret  the full spark property from the zero distribution perspective.
For a  function  system $\Lambda=\{g_{0}, \ldots, g_{L-1}\}$, its span space is defined as
 \begin{align}\label{span}\begin{array}{lllllllllllllllll}
 \hbox{span}\{\Lambda\}:=\big\{\sum^{L-1}_{j=0} c_{j}g_{j}: c_{j}\in \mathbb{R}\big\}.
 \end{array}\end{align}
%
It is easy to check that  the     full spark  property of the matrix  in \eqref{fullspark}  is equivalent to that  the function system \begin{align}\label{jihe}\begin{array}{lllllllllllllllll} \Lambda_{\varphi}:=\{\varphi, \ldots, \varphi(\cdot+s-1)\} \end{array}\end{align}
satisfies   the $(s-1)$-Haar condition (HC for short) on $(0,1)$ (c.f.\cite{Chui,DU,Van de Vel} for HC).  Specifically, $\Lambda_{\varphi}$
is linearly independent and
\begin{align}\label{quantity} \sup_{0\not\equiv h\in\hbox{span}\{\Lambda_{\varphi}\}}\#(\mathcal{Z}_{h}\cap(0,1))\leq s-1, \end{align}
where $\mathcal{Z}_{h}$ is  the zero set of $h$ and
$\#(\mathcal{Z}_{h}\cap(0,1))$
is the cardinality of   $\mathcal{Z}_{h}\cap(0,1)$.

Motivated by the above HC, from the zero distribution perspective  we will   establish
the condition on    $\phi:=\phi_{\Re}+\textbf{i}\phi_{\Im}$ such that the  PLS     in  $V_{\hbox{ca}}(\phi)$ can be achieved.
Inspired  by  Theorem  \ref{examm}, the zero distribution should not be correlated with the  functions in
$\hbox{span}\{\phi, \ldots, \phi(\cdot+s-1)\}$.
Instead we will require  in  section \ref{section2}  that
the   distribution  is related  with  the functions  in $\hbox{span}(\Xi_{\phi})$, where
\begin{align}\label{gammaset}\begin{array}{lllllllllllllllll}  \Xi_{\phi}:=\big\{\phi_{\Re}\phi_{\Re}(\cdot+k)+\phi_{\Im}\phi_{\Im}(\cdot+k),
\phi_{\Re}\phi_{\Im}(\cdot+k)- \\ \ \ \ \quad \quad \phi_{\Im}\phi_{\Re}(\cdot+k)\big\}^{s-1}_{k=1} \cup\big\{\phi^{2}_{\Re}+\phi^{2}_{\Im}\big\}.\end{array}\end{align}
More specifically,  $\Xi_{\phi}$
is  linearly independent and
\begin{align}\label{quantity11} \sup_{0\not\equiv h\in \hbox{span}\{\Xi_{\phi}\}}\mu(\mathcal{Z}_{h}\cap(0,1))=0,  \end{align}
where $\mu$ is the Lebesgue measure and $\hbox{span}\{\Xi_{\phi}\}$ is defined via  \eqref{span}.
Clearly, \eqref{quantity11} (a measure perspective) is essentially different from \eqref{quantity} (a cardinality  perspective). Compared with the cardinality perspective,  we will profit more from the measure  perspective. Details on this will be given in section \ref{yuyuou}.
For simplicity  we give the following definitions.

\begin{definition}\label{CHCcomplex}
If \eqref{quantity11} holds,  then we say  that the  system $\Xi_{\phi}$ satisfies the  \emph{generalized Haar condition} (GHC for short), and  $\phi$  is  a complex-valued  GHC-generator.
\end{definition}

As a counterpart of Definition \ref{CHCcomplex}, we next define the GHC related to  the  real-valued generator  $\varphi$   in \eqref{support}.

\begin{definition}\label{CHCreal}
If $\Lambda_{\varphi}=\{\varphi(\cdot+k): k=0, \ldots, s-1\}$ in \eqref{jihe}
is  linearly independent and satisfies
\begin{align}\label{quantity110} \sup_{0\not\equiv h\in \hbox{span}\{\Lambda_{\varphi}\}}\mu(\mathcal{Z}_{h}\cap(0,1))=0,  \end{align}
then we say that the  system   $\Lambda_{\varphi}$ satisfies the GHC,  and     $\varphi$  is  a real-valued GHC-generator.
\end{definition}

\subsection{Typical  GHC-generators}
\subsubsection{Typical complex-valued GHC-generators}\label{ghvcd87}
We start with the amplitude-phase  form of a complex-valued function.
Any function (including a generator for an SIS) $F: \mathbb{R}\longrightarrow \mathbb{C}$ can be
written as the \emph{general}  form $|F(t)|e^{\textbf{i}\theta(F(t))}$, where $|F(t)|$ and
$\theta(F(t))$ (taking   values on $[0, 2\pi)$) are referred to as the amplitude function and the  phase (or rotation) function.
Therefore,  any  complex-valued generator for an SIS  can be interpreted as the (possibly nonlinear) rotation of a real-valued function.

As mentioned before, one of  the  reasons  for PR in optics  is the high  oscillation of a signal.
Chirps   which  take   the very  \emph{general} form $F(t)e^{\textbf{i}\lambda\rho(t)}$
are the typical class  of highly oscillatory signals, where $F(t)\geq0$ and
$\lambda$ is a (large) base frequency  such that the phase function $\lambda\rho(t)$ is varying rapidly over time (c.f.\cite{CHEN}). As stated in \cite{CHEN},
chirps are \emph{ubiquitous in nature}.  They  are of  interest in applications such as
in
analysis of echolocation in bats (\cite{Science1}) and whales (\cite{Science2,CHEN}), and in  detecting gravitational waves (\cite{CHEN}).
They are also applied
in  ultrafast optics (\cite{Science3})
and ultrashort laser pulses (\cite{Science4}).


Many chirps such as those in \cite[section 6.3]{chirp} have the local analytic structure.
Employing this, GHC \eqref{quantity11} can be easily checked.
For example, motivated by \cite[section 6.3]{chirp} we take  the chirp-generator
\begin{align}\begin{array}{lll}
\label{89}\phi(x)=\frac{2}{3}\sqrt{2\pi|b|}e^{-\textbf{i}\frac{a(x-2)^{2}}{2b}}e^{-\textbf{i}\frac{p(x-2)}{b}}\cos^{2}\frac{\pi(x-2)}{4}
\chi_{(0,4)}(x),\end{array}\end{align}
such that $\hbox{supp}(\phi)=(0,4)$,
where $a\neq0$ and $\chi_{E}$ is the characteristic function of  the set $E\subseteq \mathbb{R}$.
Clearly,      the $7$ components of the    system  $\Xi_{\phi}$ in \eqref{gammaset}
are  essentially the restrictions of analytic functions.  Recall that the zero set of any nonzero  analytic function has Lebesgue
measure zero (c.f. \cite{BAF}). Hence,  if the  components $g_{i}\in\Xi_{\phi}$, $i=1, \ldots, 7$, are  linearly independent on $(0,1)$
then $\phi$ is a complex-valued GHC-generator. The  independence can be achieved if there exists  $(x_{1}, \ldots,  x_{7})\in (0,1)^{7}$
such that the determinant
\begin{align}\label{BCA} \left|\begin{array}{ccccc}
g_{1}(x_{1}) & g_{2}(x_{1}) &  \ldots& g_{7}(x_{1})\\
\vdots& \vdots & \ddots&\vdots\\
g_{1}(x_{7}) & g_{2}(x_{7}) &  \ldots& g_{7}(x_{7})
\end{array}\right|\neq 0.\end{align}
Take the parameters $(a,b,p)=(4,0.8,1)$ or $(50,0.8,1)$ for example.
Uniformly choosing $(x_{1}, \ldots,  x_{7})$ from $ (0,1)^{7}$ we found that
\eqref{BCA} holds with probability $1.$ Then  $\phi$ in \eqref{89} is a complex-valued   GHC-generator.


\subsubsection{Real-valued GHC-generators}\label{jianzhijiushishengli}
If a real-valued generator has the local analytic structure, then  we can also address how to check whether it
is a GHC-generator  by the similar argument as above.
%
%
%
The cardinal B-splines  and the refinable functions (\cite{Goodman1,Goodman2,Youfa}) having positive masks are the typical examples of real-valued GHC-generators.


\subsection{Contributions}\label{contribiution}
An infinite  discrete set $E$ is said to have  \emph{sampling density} (SD) if $\hbox{SD}=\lim_{b-a\rightarrow\infty}\frac{\#([a,   b]\cap E)}{b-a}<\infty$.
Throughout this paper, we require that  the random  sampling points on any unit interval $[n, n+1]$ obey the uniform  distribution.
Our contributions include:

(i)  From a new perspective---zero distribution,  we establish the condition for the PLS of causal signals in the complex-generated SISs.
Specifically, Theorem \ref{main1} will state  that  if $\phi$ is a complex-valued  GHC-generator,
then     with probability $1$  the random  SD $=3$ is sufficient for the PLS of   nonseparable     signals in   $V_{\hbox{ca}}(\phi)$.

(ii) The PLS of nonseparable and real-valued signals in real-generated SISs is also investigated  from the zero distribution perspective.
Specifically, Theorem \ref{main4} will state   that
if $\varphi$ is a real-valued  GHC-generator, then   with probability $1$  the random  SD $=2$ is sufficient for  the PLS of    nonseparable and
real-valued   signals in $V_{\hbox{ca}}(\varphi)$.

(iii) An alternating approach,   termed as phase decoding-coefficient recovery (PD-CR),
is established to reconstruct  the nonseparable  signals in  $V_{\hbox{ca}}(\phi)$ and $V_{\hbox{ca}}(\varphi)$.
By the random sampling-based PD-CR, Propositions \ref{main1proposition} and \ref{main1proposition123} will  guarantee that  the highly oscillatory signals can be locally  reconstructed by using  a  very small number of samples.
More details about this is given  in section \ref{yuyuou}-2).

\subsection{Highlights}\label{yuyuou}

\subsubsection{The  zero distribution perspective   enables us  to do   PLS in complex-generated $V_{\hbox{ca}}(\phi)$}
As mentioned in section \ref{newperspective}, the traditional requirement---full spark property for PLS of real-valued  signals
can be interpreted by  \eqref{quantity}, a cardinality perspective.
Theorem \ref{examm} implies that  the property
does not   work  for complex-valued case. Based on  GHC (a measure perspective), we establish the PLS
of   nonseparable  signals in  $V_{\hbox{ca}}(\phi)$.

\subsubsection{Local reconstruction of highly oscillatory signals costs a small number of  samples}\label{fgfgfa}
If $\varphi$  or  $\phi$    is     highly oscillatory, then the quantities
$ \sup_{0\not\equiv  h\in \hbox{span}\{\Lambda_{\varphi}\}}\#(\mathcal{Z}_{h}\cap(0,1))$
and $ \sup_{0\not\equiv  h\in \hbox{span}\{\Xi_{\phi}\}}\#(\mathcal{Z}_{h}\cap(0,1))$
are great.
 And a great number of deterministic samples
are necessary for   local reconstruction. However,   Propositions \ref{main1proposition} and \ref{main1proposition123} will guarantee   that  the highly oscillatory signals
can be locally  reconstructed,  with probability $1$, by using  a very small number of random samples. Unlike \cite{Sun1},   the  number
of samples is independent of  the above quantities.
 To make this point, we will  give a test signal   in section \ref{chirp} \eqref{targetnumer2}  and its local restriction in  \eqref{bianhuakuai}. Although  the   restriction  is determined by just \textbf{two} coefficients, one needs  at least $\textbf{259}$ deterministic  samples to reconstruct it.
By the PD-CR, however,
with probability $1$ it    can be reconstructed  by just \textbf{three} random samples.
\subsection{Organization}
Section  \ref{section2} concerns on  the random PLS
of    nonseparable signals  in   $V_{\hbox{ca}}(\phi)$, where $\phi$ is a   complex-valued GHC-generator. Based on the proposed  PD-CR,
we  proved that when the sampling points obey the uniform distribution and
the random SD  $=3$, then with probability $1$ any nonseparable  signal  in $V_{\hbox{ca}}(\phi)$ can be determined  up to   a unimodular scalar. In section \ref{section3} the  PD-CR  is modified such that it is more  adaptive to   the real-valued
case. By the  modified  PD-CR, the real-valued and  nonseparable signals  in  $V_{\hbox{ca}}(\varphi)$ can be determined
with probability $1$ if  the  random  SD $=2.$ To confirm our results numerical simulations are
conducted in section \ref{numeri1} and  section \ref{chirp}. For the local reconstruction, Propositions \ref{main1proposition} and \ref{main1proposition123} imply  that  the highly oscillatory signals
can be   determined,  with probability $1$, by using  a very small number of random samples.
We conclude in section \ref{conclusion}.

\section{Random PLS of  causal signals  in  complex-generated  SISs}\label{section2}
We start  with some necessary denotations.
%
As in section \ref{section1} the conjugate of $a\in \mathbb{C}$ is denoted by $\bar{a}.$
The  random  variable $t$,   obeying   the   uniform distribution on $(0,1)$, is denoted by $t\sim\textbf{U}(0,1)$.
Its observed value  is denoted by $\widehat{t}.$
For an event $\mathfrak{E}$,  its  probability and complementary event are  denoted by $P(\mathfrak{E})$  and   $\mathfrak{E}^{c}$,
respectively.
For two events $\mathfrak{E}_{1}$ and $\mathfrak{E}_{2}$, $P(\mathfrak{E}_{1}\cap \mathfrak{E}_{2})=P(\mathfrak{E}_{1}|\mathfrak{E}_{2})P(\mathfrak{E}_{2}),$
where
$\mathfrak{E}_{1}\cap \mathfrak{E}_{2}$ and  $P(\mathfrak{E}_{1}|\mathfrak{E}_{2})$ are   the  intersection event and    conditional probability, respectively.

\subsection{Preliminary on complex-valued GHC-generator}\label{GHC}
The following proposition   will be  helpful  for proving   Theorem \ref{main1}, one of our main theorems.

\begin{proposition}\label{Haarcond}
Let $\phi=\phi_{\Re}+\textbf{i}\phi_{\Im}$ be a GHC-generator supported on $(0, s)$.
Then         $\Lambda_{\phi,1}:=\{\phi(\cdot+k): k=0, \ldots, s-1\}$   and
$\Lambda_{\phi,2}:=\{\phi\bar{\phi}(\cdot+k): k=0, \ldots, s-1\}$ also  satisfy  the GHC, namely,
\eqref{quantity11} holds with $\Xi_{\phi}$ being replaced by $\Lambda_{\phi,1}$ or  $\Lambda_{\phi,2}$.
\end{proposition}
\begin{proof}
The proof can be easily  concluded by  the GHC  in \eqref{quantity11} associated with $\Xi_{\phi}$.
\end{proof}

\begin{note}
Theorem \ref{examm} implies  that      $\Lambda_{\phi,1}$  satisfying GHC
is not  sufficient for achieving  the PLS  in $V_{\hbox{ca}}(\phi)$.
By the same argument as in  the proof of Theorem \ref{examm},
it is easy to prove  that  $\Lambda_{\phi,2}$ is also not sufficient.
\end{note}

\subsection{Phase decoding-coefficient recovery   for   $V_{\hbox{ca}}(\phi)$}\label{main21}
As in section \ref{section1},  $V_{\hbox{ca}}(\phi)$ is defined  by
\begin{align}\notag \begin{array}{lll}  V_{\hbox{ca}}(\phi)=\{\sum_{k=0}^{\infty}c_k\phi(\cdot-k): \{c_{k}\}\in \ell^2, \ c_{0}\neq0, c_{k}\in \mathbb{C}\}.\end{array}\end{align}
For $f\in V_{\hbox{ca}}(\phi)$, call $\mathcal{N}_{f}:=\sup\{k: c_{k}\neq0\}$   the \emph{maximum coefficient index} of $f$. Clearly, if $f$ is compactly supported (infinitely supported) then
$\mathcal{N}_{f}<\infty(=\infty)$.
On the other hand, if the phases  of  sufficiently many samples   of $|f|$ have been  decoded, then the reconstruction of $f$ can be   linear.
Motivated by this, we will establish an alternating approach:
\emph{phase decoding-coefficient recovery} (\textbf{PD-CR}). Some denotations are necessary for introducing the approach.

As in section \ref{section1},      $\phi$ is  supported on   $(0,s)$ with the  integer $s\geq2$.
For   $n\geq1$, define the   set  $I_{n}$ by
\begin{align}\label{indexset} I_{n}:=\left\{\begin{array}{lll} \{0, 1, \ldots, n-1\},&1\leq n\leq s-1,\\
\{n-s+1,   \ldots, n-1\},&n\geq s.
\end{array}\right.\end{align}
For   $f=\sum_{k=0}^{\infty}c_k\phi(\cdot-k)\in V_{\hbox{ca}}(\phi)$,  define  the  auxiliary  function
$v_{n,f}$ on $(0,1)$ by
\begin{align} \label{vn}\begin{array}{lll}  v_{n,f}(\cdot):=\sum_{k\in I_{n}}c_{k}\phi(n+\cdot-k),\end{array}\end{align}
which together with $\hbox{supp}(\phi)\subseteq(0, s)$ leads to
  \begin{align} \label{diedai} f(n+x)=v_{n,f}(x)+c_{n}\phi(x), x\in (0, 1).\end{align}
Based on $v_{n,f}$,    define two  auxiliary bivariate  functions
\begin{align}  \label{Teren}  \begin{array}{lll} A_{n,f}(x,y)+B_{n,f}(x,y)\textbf{i}\\
:=\frac{|f|(n+x)}{|\phi|^{2}(x)}\big[\bar{\phi}(x)\phi(y)\bar{v}_{n}(y)-\bar{v}_{n}(x)|\phi|^{2}(y)\big],\end{array}\end{align}
and
\begin{align}\label{Cconstant} \begin{array}{lll}  C_{n,f}(x,y)\\
:=|f|^{2}(n+y)-|v_{n,f}|^{2}(y)+2\Re(\frac{v_{n,f}(x)\bar{v}_{n,f}(y)\phi(y)}{\phi(x)})\\
-\frac{|\phi|^{2}(y)}{|\phi|^{2}(x)}[|f|^{2}(n+x)+|v_{n,f}|^{2}(x)],\end{array}\end{align}
whenever   $x, y\in (0,1)$ such that $\phi(x)\neq0$.
The values of the  above bivariate  functions at $(x, y)\in (0, 1)^{2}$  are correlated via  the following  equation w.r.t the unknown  $z\in \mathbb{C}$:
\begin{align}\label{root} \begin{array}{lll} (A_{n,f}(x,y)+B_{n,f}(x,y)\textbf{i})z^{2}-C_{n,f}(x,y)z\\
+A_{n,f}(x,y)-B_{n,f}(x,y)\textbf{i}=0.\end{array}\end{align}
The following lemma states that the  solutions to   \eqref{root}
can provide a precise feedback on the global phase   of $\{c_{k}\}_{k\in I_{n}}$.


\begin{lemma}\label{globalphase}
Let $v_{n,f}(\cdot)$ and $A_{n,f}(\cdot,\cdot)+B_{n,f}(\cdot,\cdot)\textbf{i}$ be defined   in \eqref{vn} and  \eqref{Teren}, respectively. Define $\widetilde{v}_{n,f}(\cdot)$ via \eqref{vn} with $\{c_{k}\}_{k\in I_{n}}$ being replaced  by $\{\widetilde{c}_{k}\}_{k\in I_{n}}:=e^{\textbf{i}\widehat{\theta}}\{c_{k}\}_{k\in I_{n}}$. Moreover,
define  $\widetilde{A}_{n,f}(\cdot,\cdot)+\widetilde{B}_{n,f}(\cdot,\cdot)\textbf{i}$ and $\widetilde{C}_{n,f}(\cdot,\cdot) $ via \eqref{Teren} and \eqref{Cconstant} with
$v_{n,f}(\cdot)$ being replaced by $\widetilde{v}_{n,f}(\cdot)$.
For  fixed $x, y\in (0,1)$ and  $n\geq1$ such that $\phi(x)\neq0$ and $A_{n,f}(x,y)+B_{n,f}(x,y)\textbf{i}\neq0$,  suppose that the two solutions to    \eqref{root}
are $z_{1}$ and $z_{2}$. Then the two solutions    to
\begin{align}\label{harmonic} \begin{array}{lll} (\widetilde{A}_{n,f}(x,y)+\widetilde{B}_{n,f}(x,y)\textbf{i})z^{2}-\widetilde{C}_{n,f}(x,y)z\\
+\widetilde{A}_{n,f}(x,y)-\widetilde{B}_{n,f}(x,y)\textbf{i}=0\end{array}\end{align}
are $e^{\textbf{i}\widehat{\theta}}z_{1}$ and $e^{\textbf{i}\widehat{\theta}}z_{2}$.
\end{lemma}
\begin{proof}
Through the  direct calculation   we  have $\frac{v_{n,f}(x)\bar{v}_{n,f}(y)\phi(y)}{\phi(x)}=\frac{\widetilde{v}_{n,f}(x)\bar{\widetilde{v}}_{n,f}(y)\phi(y)}{\phi(x)}$ and
\begin{align}\label{zheqi} \begin{array}{lll} C_{n,f}(x,y)=\widetilde{C}_{n,f}(x,y),  \frac{\widetilde{A}_{n,f}(x,y)+\widetilde{B}_{n,f}(x,y)\textbf{i}}{A_{n,f}(x,y)+B_{n,f}(x,y)\textbf{i}}
=e^{-\textbf{i}\widehat{\theta}}.\end{array}\end{align}
On the other hand,
 \begin{align}\begin{array}{lll} \notag  z_{1}, z_{2}\\
 =\frac{C_{n,f}(x,y)\pm\sqrt{C^{2}_{n,f}(x,y)-4|A_{n,f}(x,y)+B_{n,f}(x,y)\textbf{i}|^{2}}}{2(A_{n,f}(x,y)+B_{n,f}(x,y)\textbf{i})},\end{array}\end{align}
which together with \eqref{zheqi} leads to that the two solutions to \eqref{harmonic} are $e^{\textbf{i}\widehat{\theta}}z_{1}$ and $e^{\textbf{i}\widehat{\theta}}z_{2}$.
\end{proof}

Based on Lemma \ref{globalphase}, the following  theorem concerns on  a   guarantee   for   decoding   phases.

\begin{theorem}\label{theorem1}
Let     $f\in V_{\hbox{ca}}(\phi)$.
Assume that    all  the    samples   \begin{align}\label{phaselesssampling}\begin{array}{lll} \{|f(\widehat{t}_{0})|\}\cup\big\{|f(n+\widehat{t}_{n_{j}})|: j=1, 2, 3, n=1, \ldots, \\
\quad\quad\quad\quad\quad\quad\quad\quad\quad\quad\quad \  \mathcal{N}_{f}+s-1\big\}\end{array}\end{align} are nonzeros  where $\widehat{t}_{0}, \widehat{t}_{n_{j}}\in (0, 1).$
Then  the corresponding  phases $\{\theta(f(\widehat{t}_{0}))\}\cup \{\theta(f(n+\widehat{t}_{n_{j}}))\}_{n, j}$    can be determined (up to a global real-valued  number) if
 for every  $n\in \{1,2, \ldots, \mathcal{N}_{f}+s-1\}$,   $\phi(\widehat{t}_{n_{1}})\neq0$ and  the  equation system
\begin{align}\label{equxy0}\left\{\begin{array}{lll} (A_{n,f}(\widehat{t}_{n_{1}},\widehat{t}_{n_{2}})+B_{n,f}(\widehat{t}_{n_{1}},\widehat{t}_{n_{2}})\textbf{i})z^{2}-C_{n,f}(\widehat{t}_{n_{1}},\widehat{t}_{n_{2}})z\\
+A_{n,f}(\widehat{t}_{n_{1}},\widehat{t}_{n_{2}})-B_{n,f}(\widehat{t}_{n_{1}},\widehat{t}_{n_{2}})\textbf{i}=0,\\
(A_{n,f}(\widehat{t}_{n_{1}},\widehat{t}_{n_{3}})+B_{n,f}(\widehat{t}_{n_{1}},\widehat{t}_{n_{3}})\textbf{i})z^{2}-C_{n,f}(\widehat{t}_{n_{1}},\widehat{t}_{n_{3}})z\\
+A_{n,f}(\widehat{t}_{n_{1}},\widehat{t}_{n_{3}})-B_{n,f}(\widehat{t}_{n_{1}},\widehat{t}_{n_{3}})\textbf{i}=0,
\end{array}\right.\end{align}
 w.r.t the  the unknown $z\in \mathbb{C}$ has a unique solution.
\end{theorem}
\begin{proof}
As previously,
denote $f=\sum_{k=0}^{\infty}c_k\phi(\cdot-k)$.
We prove the  theorem recursively on $n$.
Suppose that       \begin{align} \label{assumption0} \theta(f(\widehat{t}_{0}))=\theta_{0}\end{align}
is known as the   priori information.
 Then  it follows from $0\neq f(\widehat{t}_{0})=c_{0}\phi(\widehat{t}_{0})$ that
   \begin{align} \label{rongshui} c_{0}=e^{\textbf{i}\theta_{0}}|f(\widehat{t}_{0})|/\phi(\widehat{t}_{0}).\end{align}
For $n=1$, we next  address how to determine $z:=e^{\textbf{i}\theta(f(\widehat{t}_{1_{1}}+1))}$.
 It follows from    \eqref{diedai}   that
\begin{align}\label{iterative00} \left\{\begin{array}{lll}
|v_{1,f}(\widehat{t}_{1_{1}})+c_{1}\phi(\widehat{t}_{1_{1}})|=|f(1+\widehat{t}_{1_{1}})|,\\
|v_{1,f}(\widehat{t}_{1_{2}})+c_{1}\phi(\widehat{t}_{1_{2}})|=|f(1+\widehat{t}_{1_{2}})|,\\
|v_{1,f}(\widehat{t}_{1_{3}})+c_{1}\phi(\widehat{t}_{1_{3}})|=|f(1+\widehat{t}_{1_{3}})|,
\end{array}\right.\end{align}
where $v_{1,f}(\widehat{t}_{1_{j}}), j=1,2,3,$ are  computed by using \eqref{vn} and  \eqref{rongshui} as follows,
\begin{align}\label{buxmy} \begin{array}{lll} v_{1,f}(\widehat{t}_{1_{j}})=\phi(1+\widehat{t}_{1_{j}})c_{0}=\frac{\phi(1+\widehat{t}_{1_{j}})e^{\textbf{i}\theta_{0}}|f(\widehat{t}_{0})|}{\phi(\widehat{t}_{0})}.\end{array}\end{align}
By   \eqref{diedai},
$ |f(1+\widehat{t}_{1_{1}})|z=v_{1,f}(\widehat{t}_{1_{1}})+c_{1}\phi(\widehat{t}_{1_{1}})$.
Since $\phi(\widehat{t}_{1_{1}})\neq0$, we have  \begin{align}\label{c1} c_{1}=\frac{|f(1+\widehat{t}_{1_{1}})|z-v_{1,f}(\widehat{t}_{1_{1}})}{\phi(\widehat{t}_{1_{1}})},\end{align}
 which together with the last two  identities   in \eqref{iterative00}  leads to
 \begin{align}\label{sxdnew00} \begin{array}{lll} \big|v_{1,f}(\widehat{t}_{1_{j}})+\frac{|f(1+\widehat{t}_{1_{1}})|z-v_{1,f}(\widehat{t}_{1_{1}})}{\phi(\widehat{t}_{1_{1}})}\phi(\widehat{t}_{1_{j}})\big|^{2}\\
 =|f(1+\widehat{t}_{1_{j}})|^{2}, j=2,3.\end{array}\end{align}
By direct calculation, we can prove that   \eqref{sxdnew00} is  equivalent to   \eqref{equxy0} for  $n=1$. Since there exists a unique solution to
\eqref{equxy0},    $z$ can be determined, and consequently    $c_{1}$
can be  done  by \eqref{c1}.
Now    $\{\theta(f(1+\widehat{t}_{1_{l}})):  l\neq1\}$ are  determined  by $\theta(f(1+\widehat{t}_{1_{l}}))=\theta(v_{1,f}(\widehat{t}_{1_{l}})+c_{1}\phi(\widehat{t}_{1_{l}}))$ with $v_{1,f}(\widehat{t}_{1_{l}})$
 given  by \eqref{buxmy}. Suppose that
$\{\theta(f(\widehat{t}_{0}))\}\cup\{\theta(f(k+\widehat{t}_{k_{j}})): j=1, 2, 3, k=1, \ldots, n-1\}$
and $\{c_{k}\}^{n-1}_{k=0}$ have been determined, where $n<\mathcal{N}_{f}+s$. Through  the similar    procedures as above,
$z:=e^{\textbf{i}\theta(f(\widehat{t}_{n_{1}}+n))}$ can be determined by \eqref{equxy0},
and $c_{n}=[|f(n+\widehat{t}_{n_{1}})|z-v_{n,f}(\widehat{t}_{n_{1}})]/\phi(\widehat{t}_{n_{1}}).$
Then by   \eqref{diedai}, $\theta(f(n+\widehat{t}_{n_{j}}))$
can be computed. By the recursion on $n$, the phases $\{\theta(f(\widehat{t}_{0}))\}\cup\{\theta(f(n+\widehat{t}_{n_{j}})): j=1, 2, 3, n= 1,  \ldots, \mathcal{N}_{f}+s-1\}$    can be determined.

Recall that the above determination is achieved by   the priori information \eqref{assumption0}. Without this information,
now we   assign     \begin{align}\label{assumption1} \theta(f(\widehat{t}_{0}))=\widetilde{\theta}_{0},\end{align}
where $\widetilde{\theta}_{0}\in [0, 2\pi)$. We   next  prove that under this   assignment, $\widetilde{f}:=e^{\textbf{i}(\widetilde{\theta}_{0}-\theta_{0})}f=\sum^{\infty}_{k=0}\widetilde{c}_{k}\phi(\cdot-k)$
can be determined by the   samples in \eqref{phaselesssampling}, where $\widetilde{c}_{k}=e^{\textbf{i}(\widetilde{\theta}_{0}-\theta_{0})}c_{k}$. Consequently, $\{\theta(f(\widehat{t}_{0}))+\widetilde{\theta}_{0}-\theta_{0}\}\cup\{\theta(f(n+\widehat{t}_{n_{j}}))+\widetilde{\theta}_{0}-\theta_{0}: j=1, 2, 3, n=1,  \ldots, \mathcal{N}_{f}+s-1\}$    can be determined.

For   \eqref{assumption1},   through  the similar  analysis as  in  \eqref{rongshui}   we have
\begin{align} \notag  \begin{array}{lll} \widetilde{c}_{0}=\frac{e^{\textbf{i}\widetilde{\theta}_{0}}|f(\widehat{t}_{0})|}{\phi(\widehat{t}_{0})}
=e^{\textbf{i}(\widetilde{\theta}_{0}-\theta_{0})}c_{0}\end{array}\end{align}
and $\theta(\widetilde{f}(\widehat{t}_{0}))=\theta(\phi(\widehat{t}_{0})\widetilde{c}_{0})=\theta(f(\widehat{t}_{0}))+\widetilde{\theta}_{0}-\theta_{0}$.
As in Lemma \ref{globalphase}, define  $\widetilde{A}_{1,f}(x,y)+\widetilde{B}_{1,f}(x,y)\textbf{i}$ and $ \widetilde{C}_{1,f}(x,y)$  via \eqref{Teren} and \eqref{Cconstant}  with  $c_{0}$ being replaced  by $\widetilde{c}_{0}$.
Through  the similar  analysis as   in \eqref{sxdnew00},  $z:=e^{\textbf{i}\theta(\widetilde{f}(\widehat{t}_{1_{1}}+1))}$
satisfies   \begin{align}\label{equxy000}\left\{\begin{array}{lll} (\widetilde{A}_{1,f}(\widehat{t}_{1_{1}},\widehat{t}_{1_{2}})+\widetilde{B}_{1,f}(\widehat{t}_{1_{1}},\widehat{t}_{1_{2}})\textbf{i})z^{2}-\widetilde{C}_{1,f}(\widehat{t}_{1_{1}},\widehat{t}_{1_{2}})z\\
+\widetilde{A}_{1,f}(\widehat{t}_{1_{1}},\widehat{t}_{1_{2}})-\widetilde{B}_{1,f}(\widehat{t}_{1_{1}},\widehat{t}_{1_{2}})\textbf{i}=0,\\
(\widetilde{A}_{1,f}(\widehat{t}_{1_{1}},\widehat{t}_{1_{3}})+\widetilde{B}_{1,f}(\widehat{t}_{1_{1}},\widehat{t}_{1_{3}})\textbf{i})z^{2}-\widetilde{C}_{1,f}(\widehat{t}_{1_{1}},\widehat{t}_{1_{3}})z\\
+\widetilde{A}_{1,f}(\widehat{t}_{1_{1}},\widehat{t}_{1_{3}})-\widetilde{B}_{1,f}(\widehat{t}_{1_{1}},\widehat{t}_{1_{3}})\textbf{i}=0.
\end{array}\right.\end{align}
By the same  argument as    in   the proof of Lemma \ref{globalphase},  we can prove that   \begin{align}\label{guanxi}  \begin{array}{lll} \frac{\widetilde{A}_{1,f}(x,y)+\widetilde{B}_{1,f}(x,y)\textbf{i}}{A_{1,f}(x,y)+B_{1,f}(x,y)\textbf{i}}=e^{\textbf{i}(\theta_{0}-\widetilde{\theta}_{0})},
\widetilde{C}_{1,f}(x,y)=C_{1,f}(x,y),\end{array}\end{align}
which together with
\eqref{equxy0} having a unique solution  leads to that
\eqref{equxy000} also  has a unique solution. Applying
Lemma \ref{globalphase} with $\widehat{\theta}=\widetilde{\theta}_{0}-\theta_{0}$, we have $z=e^{\textbf{i}(\widetilde{\theta}_{0}-\theta_{0})}e^{\textbf{i}\theta(f(\widehat{t}_{1_{1}}+1))}$.
Consequently, $\widetilde{c}_{1}=e^{\textbf{i}(\widetilde{\theta}_{0}-\theta_{0})}c_{1}$.
Suppose that $\{\theta(f(\widehat{t}_{0}))+\widetilde{\theta}_{0}-\theta_{0}\}\cup\{\theta(f(k+\widehat{t}_{k_{j}}))+\widetilde{\theta}_{0}-\theta_{0}: j=1, 2, 3, k=1,  \ldots, n-1\}$ (or $\widetilde{c}_{k}=e^{\textbf{i}(\widetilde{\theta}_{0}-\theta_{0})}c_{k}$)
have been determined.
Define $\widetilde{A}_{n,f}+\widetilde{B}_{n,f}\textbf{i}$ and
$\widetilde{C}_{n,f}$ via \eqref{Teren} and \eqref{Cconstant}, respectively, by replacing $c_{k}$ by $\widetilde{c}_{k}$.
By Lemma \ref{globalphase} and the similar  analysis as in \eqref{guanxi}, we can prove that $\widetilde{c}_{n}=e^{\textbf{i}(\widetilde{\theta}_{0}-\theta_{0})}c_{n}$.
The rest of the proof can be easily  concluded by the  recursion on $n$.
\end{proof}

The procedures in the proof of Theorem \ref{theorem1}  for  decoding the  phases
$\{\theta(f(\widehat{t}_{0}))\}\cup\{\theta(f(n+\widehat{t}_{n_{j}}))\}_{n, j}$, up to a global real number,   are conducted recursively on $n$.
And they   alternate  with those for recovering the coefficients $\{c_{n}\}$.
Next we summarize them  to establish the PD-CR.

\textbf{Approach II-B} \label{decoder}

\textbf{Input}: Samples   $\{|f(\widehat{t}_{0})|\}\cup\{|f(k+\widehat{t}_{k_{j}})|:  j=1, 2, 3, k=1, \ldots, n\}$ where $\widehat{t}_{0}, \widehat{t}_{k_{j}}\in (0, 1)$ and $n< \mathcal{N}_{f}+s$;
initial phase $\theta(f(\widehat{t}_{0}))=\widetilde{\theta}_{0}$ and  $c_{0}:= e^{\textbf{i}\widetilde{\theta}_{0}}|f(\widehat{t}_{0})|/\phi(\widehat{t}_{0})$.

\textbf{Output}: $\{c_{k}\}^{n}_{k=0}$ and $\{\theta(f(\widehat{t}_{0}))\}\cup\{\theta(f(k+\widehat{t}_{k_{j}})):  j=1, 2, 3, k=1, \ldots, n\}$.

\textbf{Recursion assumption}: Assume that the  phases  $\{\theta(f(\widehat{t}_{0}))\}\cup\{\theta(f(k+\widehat{t}_{k_{j}})):  j=1, 2, 3, k=1, \ldots, n-1\}$
and coefficients  $\{c_{k}\}^{n-1}_{k=0}$ have been recovered.  Then    $\{\theta(f(n+\widehat{t}_{n_{j}})): j=1, 2, 3\}$ and $c_{n}$ are recovered by the following steps:


\textbf{step 1}: Compute $v_{n,f}(\widehat{t}_{n_{j}})$ via  \eqref{vn} with $j=1,2,3$. Compute $A_{n,f}(\widehat{t}_{n_{1}},\widehat{t}_{n_{2}})+B_{n,f}(\widehat{t}_{n_{1}},\widehat{t}_{n_{2}})\textbf{i}$ and
$A_{n,f}(\widehat{t}_{n_{1}},\widehat{t}_{n_{3}})+B_{n,f}(\widehat{t}_{n_{1}},\widehat{t}_{n_{3}})\textbf{i}$,
$C_{n,f}(\widehat{t}_{n_{1}}, \widehat{t}_{n_{2}})$ and $C_{n,f}(\widehat{t}_{n_{1}},\widehat{t}_{n_{3}})$ via \eqref{Teren} and \eqref{Cconstant}, respectively.

\textbf{step 2}: $\theta(f(n+\widehat{t}_{n_{1}}))$ is decoded  by  computing
\begin{align}\label{rty}\begin{array}{lll}e^{\textbf{i}\theta(f(n+\widehat{t}_{n_{1}}))}\\
=\arg\min_{z_{n,k}\in\{z_{n,1},z_{n,2}\}}\Big\{\min\{|z_{n,k}-z_{n,l}|: l=3,4\}\Big\}\end{array}\end{align}
where
\begin{equation}\notag\begin{array}{lll}
z_{n,k}\\
=\frac{C_{n,f}(\widehat{t}_{n_{1}},\widehat{t}_{n_{2}})}{2(A_{n,f}(\widehat{t}_{n_{1}},\widehat{t}_{n_{2}})+B_{n,f}(\widehat{t}_{n_{1}},\widehat{t}_{n_{2}})\textbf{i})}\\
\pm\frac{\sqrt{C^{2}_{n,f}(\widehat{t}_{n_{1}},\widehat{t}_{n_{2}})-4|A_{n,f}(\widehat{t}_{n_{1}},\widehat{t}_{n_{2}})+B_{n,f}(\widehat{t}_{n_{1}},\widehat{t}_{n_{2}})\textbf{i}|^{2}}}{2(A_{n,f}(\widehat{t}_{n_{1}},\widehat{t}_{n_{2}})+B_{n,f}(\widehat{t}_{n_{1}},\widehat{t}_{n_{2}})\textbf{i})}
\end{array}\end{equation}
with $k=1,2$,
and
\begin{align}\notag\begin{array}{lll}
z_{n,l}\\
=\frac{C_{n,f}(\widehat{t}_{n_{1}},\widehat{t}_{n_{3}})}{2(A_{n,f}(\widehat{t}_{n_{1}},\widehat{t}_{n_{3}})+B_{n,f}(\widehat{t}_{n_{1}},\widehat{t}_{n_{3}})\textbf{i})}\\
\pm\frac{\sqrt{C^{2}_{n,f}(\widehat{t}_{n_{1}},\widehat{t}_{n_{3}})-4|A_{n,f}(\widehat{t}_{n_{1}},\widehat{t}_{n_{3}})+B_{n,f}(\widehat{t}_{n_{1}},\widehat{t}_{n_{3}})\textbf{i}|^{2}}}{2(A_{n,f}(\widehat{t}_{n_{1}},\widehat{t}_{n_{3}})+B_{n,f}(\widehat{t}_{n_{1}},\widehat{t}_{n_{3}})\textbf{i})}
\end{array}\end{align}
with $l=3,4$.

\textbf{step 3}: Compute $c_{n}=\big[e^{\textbf{i}\theta(f(n+\widehat{t}_{n_{1}}))}|f(n+\widehat{t}_{n_{1}})|-v_{n,f}(\widehat{t}_{n_{1}})\big]/\phi(\widehat{t}_{n_{1}}).$
Compute $f(n+\widehat{t}_{n_{j}})$ by \eqref{diedai}, and
$\theta(f(n+\widehat{t}_{n_{j}}))=\theta\big(\frac{f(n+\widehat{t}_{n_{j}})}{|f(n+\widehat{t}_{n_{j}})|}\big)$ where  $j\neq 1$.


\subsection{Random phaseless sampling  for $V_{\hbox{ca}}(\phi)$}\label{prs}
Next we  replace  the    points $\{\widehat{t}_{0}\}\cup\{\widehat{t}_{n_{1}},\widehat{t}_{n_{2}}, \widehat{t}_{n_{3}}\}_{n}$  in Theorem \ref{theorem1} \eqref{phaselesssampling}
by  random variables, and
establish
our \emph{first main theorem} as follows.

\begin{theorem}\label{main1}
Let  $\phi=\phi_{\Re}+\textbf{i}\phi_{\Im}$  be a complex-valued  GHC-generator  such that   $\hbox{supp}(\phi)\subseteq (0,s)$   with the  integer $s\geq 2$.
Then  any  nonseparable  signal   $f\in V_{\hbox{ca}}(\phi)$  can be determined  (up to  a unimodular scalar) with    probability $1$ by the random samples
$\{|f(t_{0})|\}\cup\{|f(n+t_{n_{1}})|,|f(n+t_{n_{2}})|, |f(n+t_{n_{3}})|: n=1,  \ldots, \infty\}$,
where
 the i.i.d random  variables   $\{t_{0}\}\cup\{t_{n_{1}},t_{n_{2}}, t_{n_{3}}: n=1,  \ldots, \infty\}\sim\textbf{U}(0,1)$.
\end{theorem}

\begin{proof}
The proof is given in section \ref{proof2}.
\end{proof}

The following proposition concerns on  local reconstruction.
\begin{proposition}\label{main1proposition}
Let  $\phi=\phi_{\Re}+\textbf{i}\phi_{\Im}$ and $f$ be as in Theorem \ref{main1}.
Then for any integer $L>1$, the restriction $f_{[0, L]}$ of $f$ on $[0, L]$ can be determined (up to  a unimodular scalar) with    probability $1$,
by the random   samples $\{|f(t_{0})|\}\cup\{|f(n+t_{n_{1}})|,|f(n+t_{n_{2}})|, |f(n+t_{n_{3}})|: n=1,  \ldots, L-1\}$,
where
 the i.i.d random  variables   $\{t_{0}\}\cup\{t_{n_{1}},t_{n_{2}}, t_{n_{3}}: n=1,  \ldots, L-1\}\sim\textbf{U}(0,1)$.
\end{proposition}
\begin{proof}
The proof is given in section \ref{xcvc}.
\end{proof}

\subsection{Conjugation ambiguity does not occur for  $V_{\hbox{ca}}(\phi)$}\label{conjugationambiguity}
In this section we will compare the result  in section \ref{prs} with  that in \cite{Fienup45} which  concerns  on the conjugation phase retrieval.
We start with the definition of conjugation ambiguity (c.f. \cite{Fienup45,Fienup46}) of PR in an SIS.
The  conjugation ambiguity means  that there exists a signal  $f$ in an SIS, which is not real-valued (up to  a unimodular scalar),
such that it  is not distinguishable from its conjugation $\bar{f}$ which still lies in this SIS.
For a real-valued generator $\varphi$, it is clear  that
\begin{align}\label{ambi123} |\sum^{N}_{k=0}c_{k}\varphi(\cdot-k)|=
|\sum^{N}_{k=0}\overline{c}_{k}\varphi(\cdot-k)| \end{align}
for any sequence $\{c_{k}\}^{N}_{k=0}\subseteq \mathbb{C}.$
That is, the conjugation ambiguity is inevitable for phaseless sampling  in a real-generated SIS.
Most  recently,   C.K. Lai, F. Littmann and  E. Weber \cite{Fienup45}  investigated  the conjugate phase retrieval of complex-valued bandlimited signals, namely, to reconstruct them  up to the conjugation ambiguity.

It is easy to check that \eqref{ambi}
leads to \eqref{ambi123}. Despite all this, the following remark states that
there are some essential differences between
the  result in Theorem \ref{main1} and  that in  \cite{Fienup45}.

\begin{remark}
(1)
Unlike the conjugate phase retrieval, the conjugation ambiguity does not occur in    Theorem \ref{main1}.
Or else, suppose that  $f(x)=|f(x)|e^{\textbf{i}\rho(x)}$ and $\bar{f}(x)=|f(x)|e^{-\textbf{i}\rho(x)}$ both lie in $V_{\hbox{ca}}(\phi)$.
Clearly, $|f(x)|=|\bar{f}(x)|$. Then it follows from Theorem \ref{main1} that $f(x)=e^{\textbf{i}c}\bar{f}(x)$ which leads to
$\rho(x)\equiv c/2.$ This is a contradiction with the definition of conjugation ambiguity.
(2) Our generator is complex-valued
and compactly supported while the generator in  \cite{Fienup45} is real-valued and not compactly-supported.
(3) Our sampling   is random while  that in \cite{Fienup45} is deterministic.
\end{remark}

\subsection{Some lemmas for proving Theorem \ref{main1}}\label{lemmamaa}
%
%
%
%
We start with     the so called  maximum  gap.

\begin{definition}\label{gap}
For
$f=\sum^{\infty}_{k=0} c_{k}\phi(\cdot-k)\in V_{\hbox{ca}}(\phi)$,
its maximum gap  is defined as
\begin{align}\notag
\mathcal{G}_{f}=\left\{\begin{array}{lll} \max\Big\{1\leq\gamma<\infty:   \exists i\geq1  \ \hbox{s.t.} \ c_{i+\gamma}\neq0,\\
\quad \quad\quad c_{i}=\ldots=c_{i+\gamma-1}=0 \Big\},  \  \hbox{if} \ \exists c_{k}=0,\\
0, \quad \quad \quad \quad \quad \quad \quad \quad   \hbox{else}.
\end{array}\right.\end{align}
\end{definition}

The following lemma concerns on the relationship
between the    maximum gap   and   nonseparability.

\begin{lemma}\label{nonsp}
If     $f=\sum^{\infty}_{k=0}c_{k}\phi(\cdot-k)\in V_{\hbox{ca}}(\phi)$  is nonseparable,  then
$\mathcal{G}_{f}< s-1$.
\end{lemma}
\begin{proof}
Suppose that     $0=c_{i}= \ldots=c_{i+L-1}$ with $i\geq1$ and $L\geq s-1$. Define
$f_{1}=\sum_{k=0}^{i-1}c_k\phi(\cdot-k)$ and $f_{2}=\sum_{k=i+L}^{\infty}c_k\phi(\cdot-k)$.
It is easy to derive from $c_{0}\neq0$, $c_{i+L}\neq0$
and $\phi$ being a GHC-generator that $f_{1}\not\equiv 0,  f_{2}\not\equiv 0.$
Now  $\hbox{supp}(\phi)\subseteq (0,s)$ leads to   $f=f_{1}+f_{2}$ and $f_{1}f_{2}\equiv0$. That is, $f$ is separable.
This is a contradiction.
\end{proof}

The following lemma concerns on the zero property  of $A_{n,f}(x,y)+B_{n,f}(x,y)\textbf{i}$.
\begin{lemma}\label{theoremX}
Let  $\phi=\phi_{\Re}+\textbf{i}\phi_{\Im}$  be a complex-valued  GHC-generator  such that   $\hbox{supp}(\phi)\subseteq (0,s)$   with the  integer $s\geq 2$. Moreover, $f\in V_{\hbox{ca}}(\phi)$ is nonseparable, and $\{t_{0}\}\cup\{t_{n_{1}},t_{n_{2}}, t_{n_{3}}: n=1,  \ldots, \mathcal{N}_{f}+s-1\}\sim\textbf{U}(0,1)$.
Then
$P\big(A_{n,f}(t_{n_{1}},t_{n_{i}})+B_{n,f}(t_{n_{1}},t_{n_{i}})\textbf{i}\neq0\big)=1$
for any
$n\in \{1,  \ldots, \mathcal{N}_{f}+s-1\}$, where   $i=2,3.$
\end{lemma}
\begin{proof}
The proof is given in section \ref{theoremXX}.
\end{proof}

Based on Lemma \ref{theoremX}, in what follows  we investigate the probabilistic behavior of the  phase $\theta(A_{n,f}(t_{n_{1}}, t_{n_{2}})$
$+\textbf{i}B_{n,f}(t_{n_{1}}, t_{n_{2}}))$.

\begin{lemma}\label{theo2.5} Let  $f, \phi$ and $\{t_{0}\}\cup\{t_{n_{1}},t_{n_{2}}, t_{n_{3}}: n=1,  \ldots, \mathcal{N}_{f}+s-1\}\sim\textbf{U}(0,1)$ be as in Lemma   \ref{theoremX}. Then for
any fixed   $\alpha\in (0, 2\pi]$,
it holds that
\begin{align} \notag P\big(\theta(A_{n,f}(t_{n_{1}}, t_{n_{2}})+\textbf{i}B_{n,f}(t_{n_{1}}, t_{n_{2}}))\neq \alpha\big)=1.\end{align}
\end{lemma}
\begin{proof}
The proof is given in section \ref{XYZ}.
\end{proof}

Based on Lemma \ref{theo2.5}, we next investigate the uniqueness of \eqref{equxy0} with $\widehat{t}_{n_{i}}$
therein being replaced by $t_{n_{i}}\sim\textbf{U}(0,1)$.

\begin{lemma}\label{auxill0}
Let  $f, \phi$ and $\{t_{0}\}\cup\{t_{n_{1}},t_{n_{2}}, t_{n_{3}}: n=1,  \ldots, \mathcal{N}_{f}+s-1\}\sim\textbf{U}(0,1)$ be as in Lemma   \ref{theoremX}.
Then for any $n\in \{1,  \ldots, \mathcal{N}_{f}+s-1\}$,  the equation system
\begin{align}\notag\left\{\begin{array}{lll} \big(A_{n,f}(t_{n_{1}},t_{n_{2}})+B_{n,f}(t_{n_{1}},t_{n_{2}})\textbf{i}\big)z^{2}-C_{n,f}(t_{n_{1}},t_{n_{2}})z\\
+A_{n,f}(t_{n_{1}},t_{n_{2}})-B_{n,f}(t_{n_{1}},t_{n_{2}})\textbf{i}=0,\\
\big(A_{n,f}(t_{n_{1}},t_{n_{3}})+B_{n,f}(t_{n_{1}},t_{n_{3}})\textbf{i}\big)z^{2}-C_{n,f}(t_{n_{1}},t_{n_{3}})z\\
+A_{n,f}(t_{n_{1}},t_{n_{3}})-B_{n,f}(t_{n_{1}},t_{n_{3}})\textbf{i}=0,
\end{array}\right.\end{align}
has only one solution with   probability $1$.
\end{lemma}
\begin{proof}
The proof is given in section \ref{XYZZZ}.
\end{proof}

Based on Lemma \ref{auxill0} and Approach  \ref{decoder},  we next prove Theorem \ref{main1}.
\subsection{Proof of Theorem \ref{main1}}\label{proof2}
By $\Lambda_{\phi,1}$ in Proposition \ref{Haarcond} satisfying GHC, we have $P(|f(n+t_{n_{1}})|\neq0)=P(|\phi(t_{n_{1}})|\neq0)=1.$
 If   the phases
$\{\theta(f(t_{0}))\}\cup\{\theta(f(n+t_{n_{i}})): n=1,  \ldots, \infty, i=1,2,3\}$
can be determined (up to a global real  number) with probability $1$, then  $\{c_{n}\}^{\infty}_{n=0}$ can be reconstructed by \eqref{diedai}
(up to a    unimodular scalar) with the same probability. On the other hand,
by $\hbox{supp}(\phi)\subseteq (0,s)$ and the definition of
$\mathcal{N}_{f}$, it is sufficient to prove that  $\{\theta(f(t_{0}))\}\cup\{\theta(f(n+t_{n_{i}})): n=1,  \ldots, \mathcal{N}_{f}+s-1, i=1,2,3\}$
can be determined (up to a  global real number) with probability $1$.

As previously, we have
$P(|\phi(t_{0})|\neq0)=1$.
Following  Approach  \ref{decoder}, let $c_{0}:=e^{\textbf{i}\widetilde{\theta}_{0}}|f(t_{0})|/\phi(t_{0})$. Then, with    probability $1$,  $c_{0}$ can be determined
up to  the  scalar $e^{\textbf{i}\widehat{\theta}}$, where $\widehat{\theta}:=\widetilde{\theta}_{0}-\theta_{0}$
with $\theta_{0}$ being the exact phase of $f(t_{0})$.
Or with probability $1$,  $\theta(f(t_{0_{i}}))$ can be determined  up to the   number  $\widehat{\theta}.$
For any $0\leq n\leq \mathcal{N}_{f}+s-1$, suppose that the phases
$\{\theta(f(t_{0}))\}\cup\{\theta(f(k+t_{k_{j}})) :k=1, \ldots, n-1,  j=1, 2,3\}$  have been determined (up to the global real number  $\widehat{\theta}$)
with    probability $1$. Correspondingly, $\{c_{k}\}^{n-1}_{k=0}$ haven been determined (up to the scalar $e^{\textbf{i}\widehat{\theta}}$)
with    probability $1$.
Now by Lemma \ref{auxill0},  Theorem  \ref{theorem1} and   Lemma \ref{globalphase}, $\theta(f(n+t_{n_{1}})+\widehat{\theta}$ can be determined
 (up to the global real number $\widehat{\theta}$) via Approach  \ref{decoder} \eqref{rty}  with    probability $1$. Then $c_{n}$ can be determined (up to the scalar $e^{\textbf{i}\widehat{\theta}}$)
with   probability $1.$
The proof can be  concluded by the recursion.

\subsection{Proof of Proposition \ref{main1proposition}}\label{xcvc}
By $\hbox{supp}(\phi)\subseteq(0,s)$ we just need to prove that  $\{c_{n}\}^{L-1}_{n=0}$  can be determined with    probability $1$, up to  a unimodular scalar. As in section \ref{proof2},  with    probability $1$, $c_{0}$ is determined by $|f(t_{0})|$, up to  a   unimodular scalar $e^{\textbf{i}\gamma}$.
Suppose that with    probability $1$, $\{c_{k}\}^{n-1}_{k=0}$ are determined  up to   $e^{\textbf{i}\gamma}$.  Then by  the same  argument as  in
section \ref{proof2}   we can prove that with the same probability, $c_{n}$ can be determined, up to $e^{\textbf{i}\gamma}$,  by $|f(n+t_{n_{i}})|, i=1,2,3$. The proof is concluded.

\subsection{Numerical simulation: random PLS of complex-valued and highly oscillatory chirps}\label{numeri1}
\begin{figure*}[htbp]
 \scriptsize
 \centering
    \begin{minipage}[b]{.1\linewidth}
    \centerline{\includegraphics[width=13.3cm, height=3.9cm]{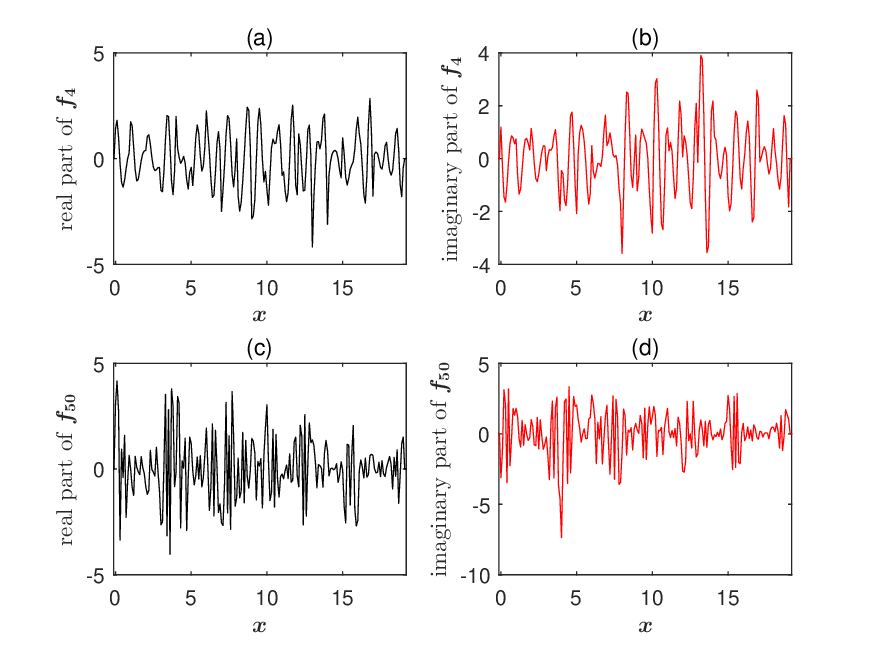}}
  \end{minipage}\hfill
 \caption{(a) The   real part of $f_{4}$; (b) The   imaginary part of $f_{4}$;
  (c) The   real part of $f_{50}$; (d) The   imaginary part of $f_{50}$.}
 \label{fig:balvsunbalHam}
\end{figure*}

 \begin{figure*}[htbp]
 \scriptsize
 \centering
    \begin{minipage}[b]{.1\linewidth}
    \centerline{\includegraphics[width=16.8cm, height=3.9cm]{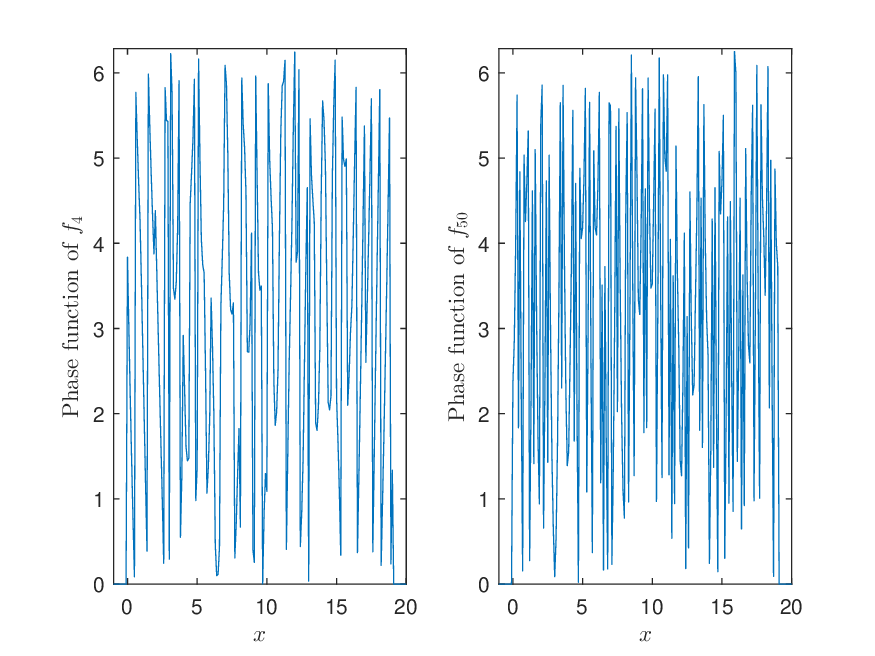}}
  \end{minipage}\hfill
 \caption{(a)  The phase function $\theta(f_{4}(x))$ of $f_{4}$; (b) The phase function $\theta(f_{50}(x))$ of $f_{50}$.}
 \label{fig:balvsunbalHam1234567}
\end{figure*}

\begin{figure*}[htbp]
 \scriptsize
 \centering
    \begin{minipage}[b]{.1\linewidth}
    \centerline{\includegraphics[width=16.8cm, height=4.1cm]{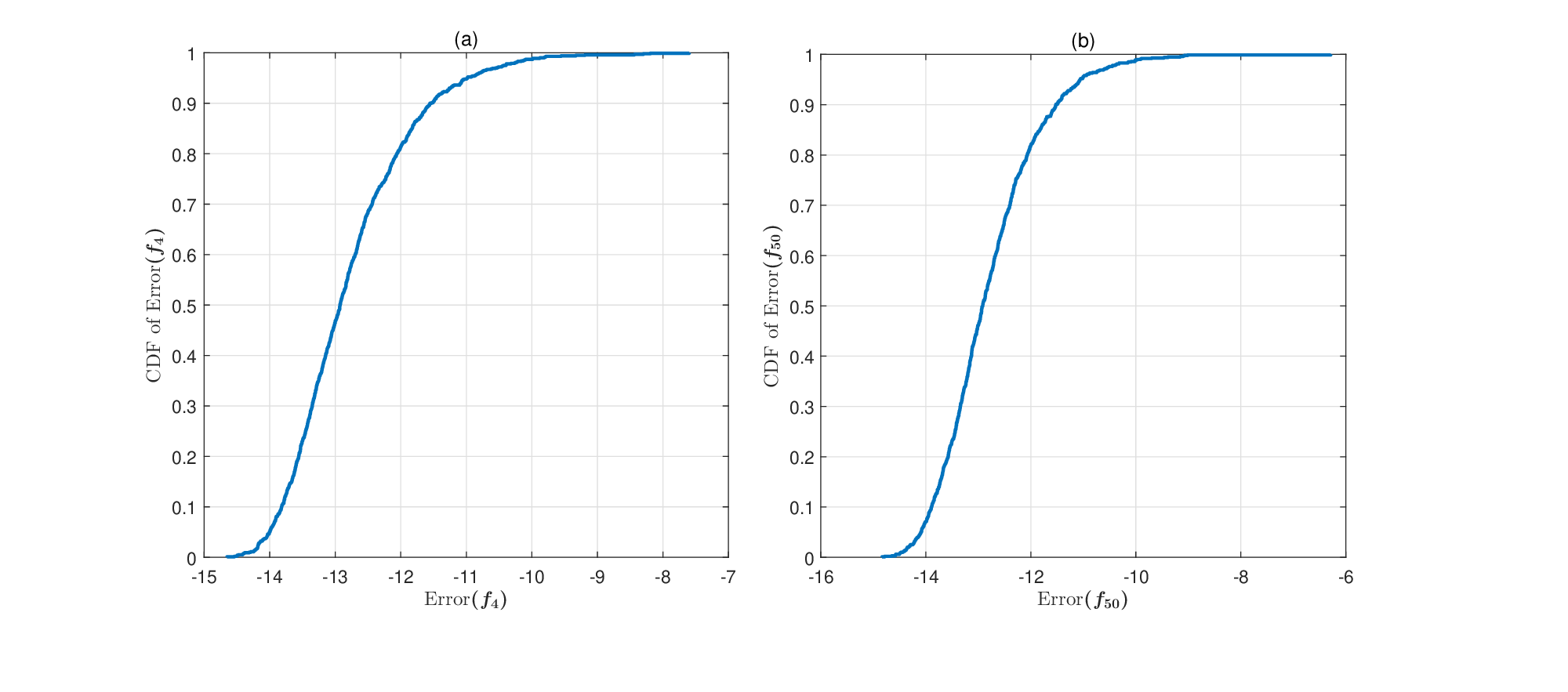}}
  \end{minipage}\hfill
 \caption{(a)  The CDF of $\hbox{Error}(f_{4})$ in the noiseless setting; (b) The CDF of  $\hbox{Error}(f_{50})$ in the noiseless setting.}
 \label{fig:balvsunbalHam123}
\end{figure*}

\begin{table*}[htbp]
 \scriptsize
 \centering
\begin{tabular}{|c|cccccccccccccc|c||} \hline \diagbox{$f_{a}$}{SNR}&50&60&70&80&90&100&110&120&130 \\
 \hline $f_{4}$&0.0070&0.0940&0.2800&0.4830&0.6440&0.7990&0.8860&0.9410&0.9970
  \\ $f_{50}$&0.0240&0.1770&0.4040& 0.6270&0.7330 &0.8430&0.9090&0.9520& 0.9880 \\
  \hline
 \end{tabular}
 \caption{Success rate  vs     noise level (SNR).}
 \label{jkk}
 \end{table*}

This  section is to verify     Theorem \ref{main1}.
Our test  SIS  $V_{\hbox{ca}}(\phi_{a,b,p})$ is related with   \cite[section 6.3.1]{chirp}. Specifically,
\begin{align}\notag \begin{array}{lll}
\phi_{a,b,p}(x)\\
=\frac{2}{3}\sqrt{2\pi|b|}e^{-\textbf{i}\frac{a(x-2)^{2}}{2b}}e^{-\textbf{i}\frac{p(x-2)}{b}}\cos^{2}\frac{\pi(x-2)}{4}\chi_{(0,4)}(x).
\end{array}\end{align}
By section \ref{ghvcd87},  both  $\phi_{4,0.8,1}$ and $\phi_{50,0.8,1}$ are GHC-generators.
The test  signal
\begin{align}\notag\begin{array}{lll}  f_{a}(x):=\sum_{n=0}^{15}c_n\phi_{a,0.8,1}(x-n),\end{array}\end{align}
where $a=4, 50$ and $\mathcal{G}(f_{a})<3$. See Fig. \ref{fig:balvsunbalHam}  for their graphs.
In Fig.  \ref{fig:balvsunbalHam1234567} we also plot the phase function  $\theta(f_{a}(x))$ defined  via  $f_{a}(x)=|f_{a}(x)|e^{\textbf{i}\theta(f_{a}(x))}$.
 Clearly, the two signals  are highly oscillatory.
By   Theorem \ref{main1},  $f_{a}(x)$  can be determined  with probability $1$, up to a unimodular,  by the
random    samples $\{|f_{a}(t_{0})|\}\bigcup\{|f_{a}(n+t_{n_{1}})|,|f_{a}(n+t_{n_{2}})|, |f_{a}(n+t_{n_{3}})|: n=1, \ldots, 18\},$ where $t_{0}, t_{n_{1}}, t_{n_{2}}, t_{n_{3}}\sim \textbf{U}(0,1)$.
In the noiseless setting, $10^3$ trials are  conducted   to determine $f_{a}(x)$ by   Approach \ref{decoder}.
The  error is defined as
\begin{align} \notag\begin{array}{lll} \hbox{Error}(f_{a})\\
:=\log_{10}(\min_{\gamma\in (0, 2\pi]}||\{c_k\}-e^{\textbf{i}\gamma}\{\widetilde{c}_k\}||_{2}/||\{c_k\}||_{2}),\end{array}\end{align}
where   $\{\widetilde{c}_k\}$ is  the coefficient sequence of the  reconstruction  version  $\widetilde{f}_{a}(x)=\sum_{n=0}^{15}\widetilde{c}_n \phi_{a,0.8,1}(x-n)$.  Approach \ref{decoder} is considered to  be successful  if $\hbox{Error}(f_{a})\leq-1.8$. The cumulative distribution function (CDF) of $\hbox{Error}(f_{a})$
is defined as
\begin{align}\label{cdf} \hbox{CDF}(x)=\frac{\# \big(\hbox{Error}(f_{a}) \leq x\big)}{10^{3}}. \end{align}
Fig. \ref{fig:balvsunbalHam123} confirms  that  with probability $1$,  the signals can be    determined  in the
noiseless setting.

In what follows we examine the robustness  of Approach \ref{decoder} to   noise corruption.
The observed values of $\{|f(t_{0})|\}\cup\{|f(n+t_{n_{1}})|,|f(n+t_{n_{2}})|, |f(n+t_{n_{3}})|: n=1,  \ldots, 18\}$
in a trial are  denoted by $\{|f(\widehat{t}_{0})|\}\cup\{|f(n+\widehat{t}_{n_{1}})|,|f(n+\widehat{t}_{n_{2}})|, |f(n+\widehat{t}_{n_{3}})|: n=1,  \ldots, 18\}$.
We add  the  Gaussian noise $\varepsilon\sim \textbf{N}(0, \sigma^{2})$ to the    noiseless samples.
That is, we employ the noisy samples  $\{|f_{a}(\widehat{t}_{0})|+\varepsilon\}\cup\{|f_{a}(n+\widehat{t}_{n_{1}})|+\varepsilon,|f_{a}(n+\widehat{t}_{n_{2}})|+\varepsilon, |f_{a}(n+\widehat{t}_{n_{3}})|+\varepsilon: n=1, \ldots, 18\}$
to conduct   Approach \ref{decoder}.
The variance $\sigma^{2}$   is chosen such that the desired signal to noise ratio  (SNR)
is expressed by
\begin{align}\label{SNRDE} \begin{array}{lllllllllllllllll} \hbox{SNR}=10\log_{10}\big(\frac{||\textbf{F}_{a}||^{2}_{2}}{55\sigma^{2}}\big),\end{array}\end{align}
where  $||\textbf{F}_{a}||^{2}_{2}=|f_{a}(\widehat{t}_{0})|^{2}+\sum^{3}_{k=1}\sum^{18}_{n=1}|f_{a}(n+\widehat{t}_{n_{k}})|^{2}$.
In the noisy setting, $10^3$  trials are also  conducted  to reconstruct $f_{4}(x)$ and $f_{50}(x)$, respectively. Their reconstruction   success rates ($\hbox{CDF}(-1.8)$) are recorded in Table \ref{jkk}.

\section{Random phaseless sampling  of  causal and real-valued  signals  in real-generated  SISs}\label{section3}
Throughout this section,
let $\varphi $ be a real-valued  GHC-generator  such that  $\hbox{supp}(\varphi)\subseteq (0, s)$ with the  integer $s\geq 2$.
This section focuses  on the PLS of   real-valued    signals in
\begin{align} \notag \begin{array}{lll} V_{\hbox{ca}}(\varphi)
=\big\{\sum^{\infty}_{k=0}c_k\varphi(\cdot-k): \{c_k\in \mathbb{R}\}\in \ell^2,  c_{0}\neq0\big\}.\end{array}\end{align}
Some denotations and definitions  are helpful  for  discussion.
Suppose that the    signal $f\in V_{\hbox{ca}}(\varphi)$
is denoted by
\begin{align}\label{targetreal}\begin{array}{lll}  f=\sum_{k=0}^{\infty}c_k\varphi(\cdot-k).\end{array}\end{align}
As in section \ref{main21} denote $\mathcal{N}_{f}=\sup\{k: c_{k}\neq0\}$.
As in \eqref{indexset},
define the  index set  $I_{n}$ by
\begin{align}\notag  I_{n}=\left\{\begin{array}{lll} \{0, 1, \ldots, n-1\},&1\leq n\leq s-1,\\
\{n-s+1,   \ldots, n-1\},&n\geq s.
\end{array}\right.\end{align}
For $n\geq1$ and  the signal   $f$ in \eqref{targetreal},  define an  auxiliary  function
\begin{align} \label{vnreal0}\begin{array}{lll}  v^{\Re}_{n,f}(x):=\sum_{k\in I_{n}}c_{k}\varphi(n+x-k), x\in (0,1).\end{array}\end{align}
Moreover, define
\begin{align} \label{vnreal}\begin{array}{lll}
A^{\Re}_{n,f}(x,y):=\frac{|f|(n+x)}{|\varphi|^{2}(x)}\Big[\varphi(x)\varphi(y)v^{\Re}_{n,f}(y)\\
\quad \quad \quad\quad \quad\quad\quad \quad\quad\quad  -v^{\Re}_{n,f}(x)\varphi^{2}(y)\Big],\end{array}\end{align}
and
\begin{align}\label{Cnx} \begin{array}{lll}  C^{\Re}_{n,f}(x,y)&:=|f|^{2}(n+y)-|v^{\Re}_{n,f}|^{2}(y)\\
&+\frac{2v^{\Re}_{n,f}(x)v^{\Re}_{n,f}(y)\varphi(y)}{\varphi(x)}\\
&\ -\frac{|\varphi|^{2}(y)}{|\varphi|^{2}(x)}\Big[|f|^{2}(n+x)+|v^{\Re}_{n,f}|^{2}(x)\Big],\end{array}\end{align}
whenever  $x, y\in (0,1)$ such that $\varphi(x)\neq0$. The maximum gap $\mathcal{G}_{f}$ is defined via Definition \ref{gap} with $\phi$ replaced by $\varphi.$ The sign function $\hbox{sgn}(x)$ takes  $1,-1$ and $0$ when $x>0, x<0$ and $x=0,$  respectively.

\begin{figure*}[htbp]\label{kncc}
 \scriptsize
 \centering
    \begin{minipage}[b]{.1\linewidth}
    \centerline{\includegraphics[width=15cm, height=7.8cm]{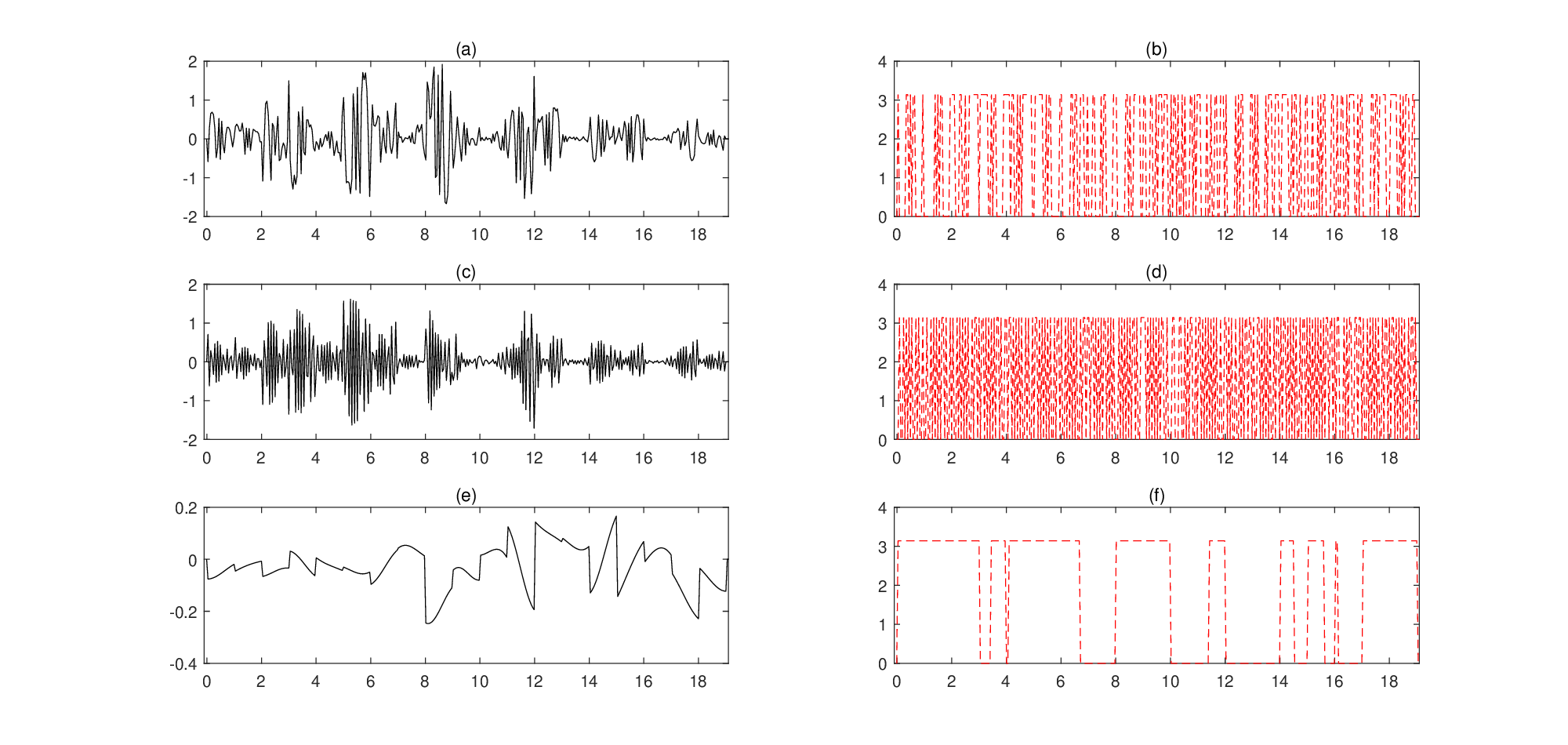}}
  \end{minipage}\hfill
 \caption{(a) The graph  of $f_{10,-0.2381,1}$; (b) The graph of the phase function of  $f_{10,-0.2381,1}$; (c) The graph  of $f_{50,-0.2381,1}$;
  (d) The graph of the phase function of $f_{50,-0.2381,1}$;
  (e) The graph  of $f_{10^{-6},-0.0038,0}$;
  (f) The graph of the phase function of $f_{10^{-6},-0.0038,0}$.}
 \label{fig:balvsunbalHam156}
\end{figure*}

\begin{figure*}[htbp]
 \scriptsize
 \centering
    \begin{minipage}[b]{.1\linewidth}
    \centerline{\includegraphics[width=16cm, height=4.1cm]{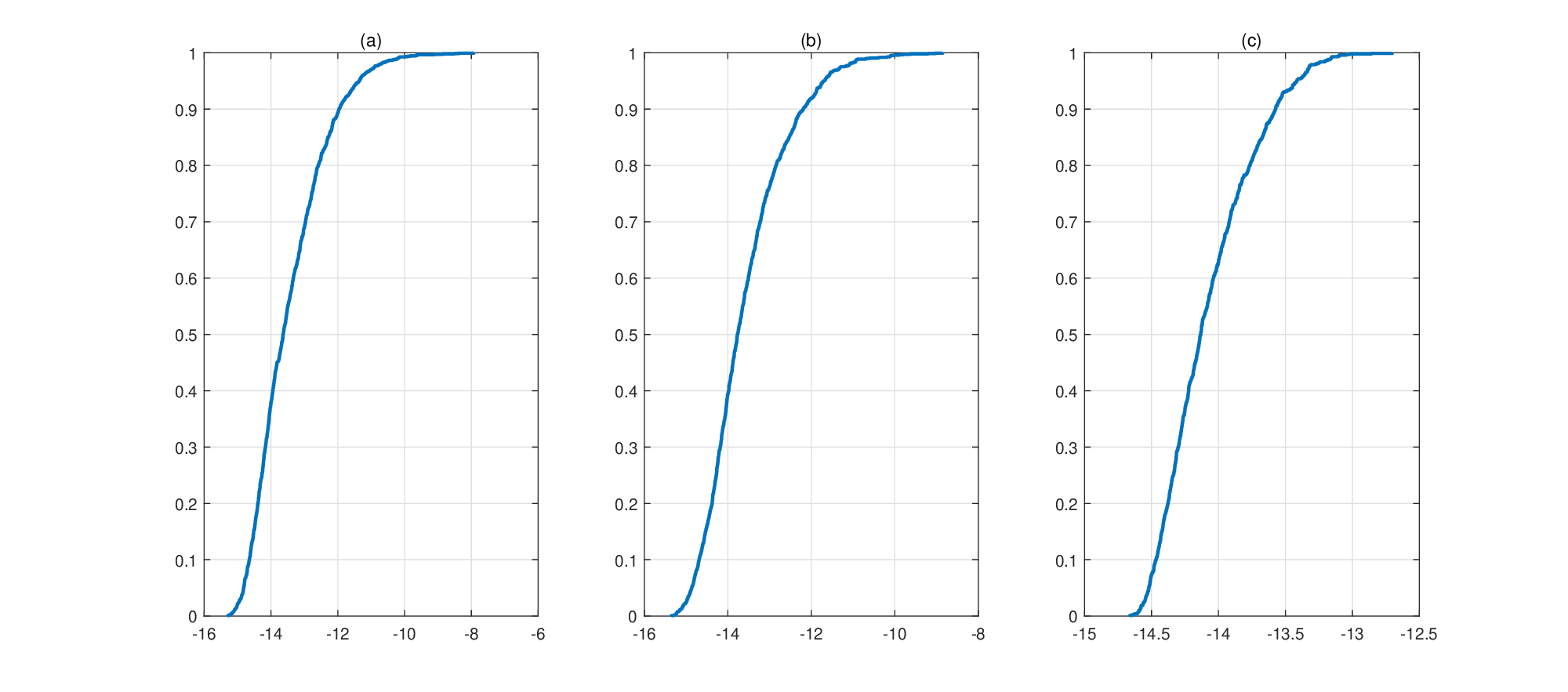}}
  \end{minipage}\hfill
 \caption{(a)  The CDF of $\hbox{Error}(f_{10,-0.2381,1})$ in the noiseless setting; (b) The CDF of  $\hbox{Error}(f_{50,-0.2381,1})$ in the noiseless setting; (c)
 The CDF of  $\hbox{Error}(f_{10^{-6},-0.0038,0})$ in the noiseless setting}
 \label{fig:balvsunbalHam167}
\end{figure*}

\begin{figure*}[htbp]
 \scriptsize
 \centering
    \begin{minipage}[b]{.1\linewidth}
    \centerline{\includegraphics[width=16cm, height=4.1cm]{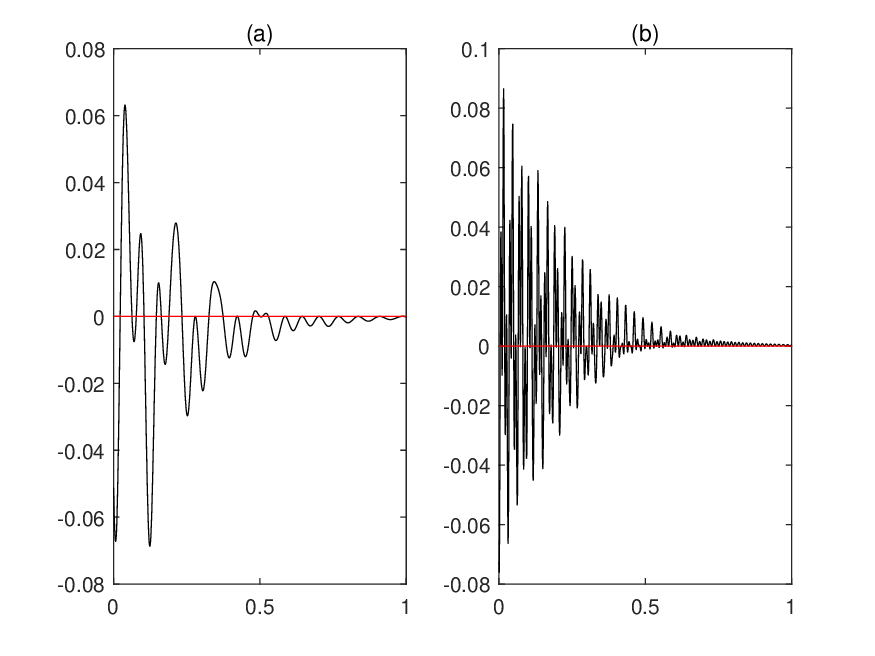}}
  \end{minipage}\hfill
 \caption{(a)  The graph of $A^{\Re}_{1,g_{10}}(0.5,x)$; (b) The graph of $A^{\Re}_{1,g_{50}}(0.5,x)$}
 \label{NND}
\end{figure*}


\subsection{Random PLS of  real-valued   signals in $V_{\hbox{ca}}(\varphi)$}
We next  establish the main theorem of  this section. It is the counterpart of Theorem \ref{main1}.

\begin{theorem}\label{main4}
Let  $\varphi$  be a real-valued  GHC-generator  such that   $\hbox{supp}(\varphi)\subseteq (0,s)$   with the  integer $s\geq 2$.
Then  any nonseparable and real-valued  signal   $f\in V_{\hbox{ca}}(\varphi)$  can be determined  (up to  a sign) with    probability $1$ by the random   samples
$\{|f(t_{0})|\}\cup\{|f(n+t_{n_{1}})|,|f(n+t_{n_{2}})|: n=1,  \ldots, \infty\}$,
where
 the i.i.d random  variables   $\{t_{0}\}\cup\{t_{n_{1}},t_{n_{2}}: n=1,  \ldots, \infty\}\sim\textbf{U}(0,1)$.
\end{theorem}
\begin{proof}
Since $f$ is nonseparable, by the     same  argument as      in   Lemma \ref{nonsp} we have
$\mathcal{G}_{f}<s-1$, which together with $\varphi$  being  a real-valued  GHC-generator leads to   the probability  $ P(|f(k+t_{k_{i}})|>0)=1$
for any  $k\in \{0,  \ldots, \mathcal{N}_{f}+s-1\}.$
Clearly, $ P(|\varphi(t_{n_{i}})|>0)=1.$
As in section \ref{proof2}, it is sufficient to prove that the phases
$\{\theta(f(t_{0}))\}\cup\{\theta(f(n+t_{n_{i}})): n=1,  \ldots, \mathcal{N}_{f}+s-1, i=1,2\}$
can be determined, up to the   constant $\pi$, with probability $1.$

 Denote $f=\sum_{k=0}^{\infty}c_k\varphi(\cdot-k). $  By  the similar argument as   in the proof of  Lemma \ref{theoremX}, we can  prove that
 \begin{align}\label{gurantee}
P(A^{\Re}_{n,f}(t_{n_{1}},t_{n_{2}})\neq0)=1, n=1,  \ldots, \mathcal{N}_{f}+s-1.
 \end{align}
Assume  that  $z_{0}=e^{\textbf{i}\theta(f(t_{0}))}\in\{1, -1\}$ is  assigned exactly. As  previously, $ P(|f(n+t_{n_{i}})|>0)=P(|\varphi(t_{n_{i}})|>0)=1.$ Then
\begin{align}\label{assume1} \begin{array}{lll} f(t_{0})=z_{0}|f(t_{0})|,\end{array}\end{align}
and $c_{0}=\frac{z_{0}|f(t_{0})|}{\varphi(t_{0})}$  with probability $1$.
We next  determine $\theta(f(t_{1_{1}}+1))$ and $c_{1}$.
Similarly to  \eqref{iterative00},
we have
\begin{align}\label{iterative01} \left\{\begin{array}{lll}
|v^{\Re}_{1,f}(t_{1_{1}})+c_{1}\varphi(t_{1_{1}})|=|f(1+t_{1_{1}})|,\\
|v^{\Re}_{1,f}(t_{1_{2}})+c_{1}\varphi(t_{1_{2}})|=|f(1+t_{1_{2}})|.
\end{array}\right.\end{align}
Denote  $f(1+t_{1_{1}}):=z_{1}|f(1+t_{1_{1}})|$ with $z_{1}\in \{1,-1\}$ to be determined.
By the similar  argument  as  in  \eqref{sxdnew00}, we can prove that  $z_{1}$ is the solution to
\begin{align}\label{positive} \begin{array}{lll} A^{\Re}_{1,f}(t_{1_{1}},t_{1_{2}})z^{2}-C^{\Re}_{1,f}(t_{1_{1}},t_{1_{2}})z+A^{\Re}_{1,f}(t_{1_{1}},t_{1_{2}})=0.\end{array}\end{align}
It follows from  \eqref{gurantee} that   with probability $1 $, there exist at most  two  solutions to the above equation.
Note  that   the product of the two solutions  is $1$. Then
there exists  a unique   solution   with the same probability.
More precisely,  \begin{align}\label{signdecoding} \begin{array}{lll} z_{1}=\hbox{sgn}\big(\frac{C^{\Re}_{1,f}(t_{1_{1}},t_{1_{2}})}{A^{\Re}_{1,f}(t_{1_{1}},t_{1_{2}})}\big).\end{array}\end{align}
Therefore under   the assumption \eqref{assume1}, $c_{1}=\frac{z_{1}|f(1+t_{1_{1}})|-v_{1,f}(t_{1_{1}})}{\varphi(t_{1_{1}})}$ with   probability $1$.
And $\theta(f(1+t_{1_{i}}))$ can be determined with probability $1$. Continuing the above procedures, $\{\theta(f(t_{0}))\}\cup\{\theta(f(n+t_{n_{i}})): n=1,  \ldots, \mathcal{N}_{f}+s-1, i=1,2\}$
can be determined  with the same  probability.

Contrary to \eqref{assume1},
we next assign \begin{align}\label{assume2} \begin{array}{lll} f(t_{0})=-z_{0}|f(t_{0})|.\end{array}\end{align}
 Under \eqref{assume2}, we shall prove that $\{\theta(f(t_{0}))+\pi\}\cup\{\theta(f(n+t_{n_{i}}))+\pi: n=1,  \ldots, \mathcal{N}_{f}+s-1, i=1,2\}$ or $\widetilde{f}=\sum_{k=0}^{\infty}\widetilde{c}_k\varphi(\cdot-k) $ can be determined with probability $1$, where $\widetilde{c}_k=-c_{k}$.
 First it follows from $P(|\varphi(t_{0})|\neq0)=1$ that  \begin{align}\label{assume3} \begin{array}{lll} \widetilde{c}_{0}=-\frac{z_{0}|f(t_{0})|}{\varphi(t_{0})}=-c_{0}.\end{array}\end{align}
 Then $\theta(\widetilde{f}(t_{0_{i}}))=\theta(f(t_{0_{i}}))+\pi$.
By \eqref{vnreal}, \eqref{Cnx} and \eqref{assume3}, we have
$A^{\Re}_{1,\widetilde{f}}(t_{1_{1}},t_{1_{2}})=-A^{\Re}_{1,f}(t_{1_{1}},t_{1_{2}})$
and   $C^{\Re}_{1,\widetilde{f}}(t_{1_{1}},t_{1_{2}})=C^{\Re}_{1,f}(t_{1_{1}},t_{1_{2}})$.
Moreover,   as  in \eqref{positive},
$\hbox{sgn}(\widetilde{f}(1+t_{1_{1}}))$ is the solution to
\begin{align} \notag \begin{array}{lll} A^{\Re}_{1,\widetilde{f}}(t_{1_{1}},t_{1_{2}})z^{2}-C^{\Re}_{1,\widetilde{f}}(t_{1_{1}},t_{1_{2}})z+A^{\Re}_{1,\widetilde{f}}(t_{1_{1}},t_{1_{2}})=0.\end{array}\end{align}
As in \eqref{signdecoding}, the solution is given by  \begin{align} \notag \begin{array}{lll} z=\hbox{sgn}\big(\frac{C^{\Re}_{1,\widetilde{f}}(t_{1_{1}},t_{1_{2}})}{A^{\Re}_{1,\widetilde{f}}(t_{1_{1}},t_{1_{2}})}\big)=-\hbox{sgn}\big(\frac{C^{\Re}_{1,f}(t_{1_{1}},t_{1_{2}})}{A^{\Re}_{1,f}(t_{1_{1}},t_{1_{2}})}\big).\end{array}\end{align}
Then $\theta(\widetilde{f}(1+t_{1_{1}}))=\theta(f(1+t_{1_{1}}))+\pi$. Consequently,  $\widetilde{c}_{1}=-c_{1}$ and
$\theta(\widetilde{f}(1+t_{1_{2}}))=\theta(f(1+t_{1_{2}}))+\pi$.
By recursion on $n$, we can prove that $\{\theta(f(t_{0}))+\pi\}\cup\{\theta(f(n+t_{n_{i}}))+\pi: n=1,  \ldots, \mathcal{N}_{f}+s-1, i=1,2\}$
can be determined  with probability $1.$
\end{proof}

The following proposition concerns on the local reconstruction.
It is the counterpart of Proposition \ref{main1proposition}.

\begin{proposition}\label{main1proposition123}
Let  $\varphi$ and $f$ be as in Theorem \ref{main4}.
Then for any integer  $L>1$, the restriction $f_{[0, L]}$ of $f$ on $[0, L]$ can be determined with    probability $1$, up to  a sign,
by the random    samples $\{|f(t_{0})|\}\cup\{|f(n+t_{n_{1}})|,|f(n+t_{n_{2}})|: n=1,  \ldots, L-1\}$,
where
 the i.i.d random  variables   $\{t_{0}\}\cup\{t_{n_{1}},t_{n_{2}}: n=1,  \ldots, L-1\}\sim\textbf{U}(0,1)$.
\end{proposition}
\begin{proof}
The proof is based on that of Theorem \ref{main4}. And it can be concluded by  the
similar argument as in  section \ref{xcvc}.
\end{proof}

\subsection{PD-CR for nonseparable real-valued signals in $V_{\hbox{ca}}(\varphi)$}
Let $\{\widehat{t}_{0}\}\cup\{\widehat{t}_{n_{1}},\widehat{t}_{n_{2}}: n=1,  \ldots, \infty\}$ be the observed values of random  variables   $\{t_{0}\}\cup\{t_{n_{1}},t_{n_{2}}: n=1,  \ldots, \infty\}$ in Theorem \ref{main4}.
Based on the proof of Theorem \ref{main4}, in what follows
we  establish  an approach for the PLS of nonseparable real-valued signals in $V_{\hbox{ca}}(\varphi)$.

\textbf{Approach III-B} \label{gghh}

\textbf{Input}: Samples   $\{|f(\widehat{t}_{0})|\}\cup\{|f(k+\widehat{t}_{k_{j}})|:  j=1, 2, k=1, \ldots, n\}$ where $\widehat{t}_{0}, \widehat{t}_{k_{j}}\in (0, 1) $ and  $n\leq \mathcal{N}_{f}+s-1$.
Assign initial phase $\theta(f(\widehat{t}_{0}))=\widetilde{\theta}_{0}\in \{0, \pi\}$; $c_{0}= e^{\textbf{i}\widetilde{\theta}_{0}}|f(\widehat{t}_{0})|/\varphi(\widehat{t}_{0})$.

\textbf{Output}: $\{c_{k}\}^{n}_{k=0}$ and $\{\theta(f(\widehat{t}_{0}))\}\cup\{\theta(f(k+\widehat{t}_{k_{j}})):  j=1, 2, k=1, \ldots, n\}$.

\textbf{Recursion assumption}: Assume that the  phases  $\{\theta(f(k+\widehat{t}_{k_{j}})):  j=1, 2, k= 1, \ldots, n-1\}$
and coefficients  $\{c_{k}\}^{n-1}_{k=0}$ have been recovered.
Then    $\{\theta(f(n+\widehat{t}_{n_{j}})): j=1, 2\}$ and $c_{n}$ are recovered by the following steps:

\textbf{step 1:} Compute $v^{\Re}_{n,f}(\widehat{t}_{n_{1}})$, $A^{\Re}_{n,f}(\widehat{t}_{n_{1}},\widehat{t}_{n_{2}})$   and $C^{\Re}_{n,f}(\widehat{t}_{n_{1}},\widehat{t}_{n_{2}})$
by \eqref{vnreal0}, \eqref{vnreal} and  \eqref{Cnx},  respectively.

\textbf{step 2:}
$\theta(f(n+\widehat{t}_{n_{1}}))\in \{0, \pi\}$
is recovered by
computing
$e^{\textbf{i}\theta(f(n+\widehat{t}_{n_{1}}))}=\hbox{sgn}\big(\frac{C^{\Re}_{n,f}(\widehat{t}_{n_{1}},\widehat{t}_{n_{2}})}{A^{\Re}_{n,f}(\widehat{t}_{n_{1}},\widehat{t}_{n_{2}})}\big)$.
And $c_{n}=[e^{\textbf{i}\theta(f(n+\widehat{t}_{n_{1}}))}|f(n+\widehat{t}_{n_{1}})|-v^{\Re}_{n,f}(\widehat{t}_{n_{1}})]/\varphi(\widehat{t}_{n_{1}})$.



\subsection{Numerical simulation: costing  small number of     samples to reconstruct  highly oscillatory real-valued chirps}\label{chirp}
 \begin{table*}[htbp]
 \scriptsize
 \centering
\begin{tabular}{|c|cccccccccccccc|c||} \hline \diagbox{$f_{a,b,p}$}{SNR}&35&40& 50&60&70&80&90&100 \\
 \hline $f_{50,-0.2381,1}$& 0.0100& 0.0590&0.2460&0.5290&0.6870&0.8450&0.9200&0.9620 \\
 $f_{10,-0.2381,1}$&0.0160&0.0470& 0.2520&0.4590& 0.6570&0.8540&0.9090&0.9560 \\
  $f_{10^{-6},-0.0038,0}$&0.1100&  0.3370& 0.7390& 0.9080 & 0.9520&  0.9800 &0.9850 &0.9960\\
  \hline
 \end{tabular}
 \caption{Success rate  vs     noise level (SNR).}
 \label{table123}
 \end{table*}


This section is to  examine    the efficiency of Approach \ref{gghh}.
The generator  $\varphi$  is chosen as $\phi_{a,b,p,\Re}$, the real part of $\phi_{a,b,p}$ defined  in    section \ref{numeri1}.
The test  signal   is
\begin{align}\label{targetnumer2}\begin{array}{lll}  f_{a,b,p}(t)=\sum_{n=0}^{15}c_{n,a}\phi_{a,b,p,\Re}(t-n), \end{array}\end{align}
where $c_{0,a}\neq0, c_{n,a}\in \mathbb{R}$.
It is easy to check that $f_{a,b,p}(t)$ can be rewritten as the real-valued  chirp   form (c.f. \cite{CHEN}): $A(t)\cos(\lambda\upsilon(t))$ with $A(\cdot)\geq0$. By the  analysis     in section \ref{jianzhijiushishengli},
we can check   that   $\phi_{10,-0.238,1,\Re}$, $\phi_{50,-0.238,1,\Re}$
and $\phi_{10^{-6},-0.0038,0,\Re}$
are all  real-valued   GHC-generators.  We choose $f_{10,-0.2381,1}$, $f_{50,-0.2381,1}$
and $f_{10^{-6},-0.0038,0}$ as test signals.
Their graphs    are plotted in Fig. \ref{fig:balvsunbalHam156}  (a, c, e). Moreover,  the phase function
$\theta(f_{a,b,p}(x))$, taking $0$ and $\pi$ when $f_{a,b,p}(x)\geq0$ and $f_{a,b,p}(x)<0$, respectively,
is plotted in Fig. \ref{fig:balvsunbalHam156}  (b, d, f).


Fig. \ref{fig:balvsunbalHam156}  (b, d, f) imply  that   $f_{10,-0.2381,1}$ and  $f_{50,-0.2381,1}$ are much more oscillatory
than $f_{10^{-6},-0.0038,0}$. It should be noted that a great number of  deterministic samples  are necessary for the local reconstructions of   $f_{10,-0.2381,1}$ and  $f_{50,-0.2381,1}$. To make this point,
define
\begin{align}\label{bianhuakuai} g_{a}(t)=\sum_{n=0}^{1}c_{n,a}\phi_{a,-0.238,1,\Re}(t-n), \ t\in (0, 2),\end{align}
where $c_{0,10}=0.7064, c_{1,10}=-0.6183, c_{0, 50}=-0.5874$ and $ c_{1,50}=0.2659$ are as in \eqref{targetnumer2}. Clearly,
\begin{align}\label{jycs} g_{a}(t)=f_{a,-0.2381,1}(t), \ t\in (0, 2),\end{align}
and the  reconstruction  of  $g_{a}$ is equivalent to ones   of   $c_{0,a}$ and $c_{1,a}$.
Suppose that      $c_{0,a}$ and $c_{1,a}$ can be recovered  by   \textbf{any}   $\mathring{L}$ deterministic samples $\{|g_{a}(\widehat{t}_{0})|,   |g_{a}(1+\widehat{t}_{1_{i}})|: i=1, \ldots, \mathring{L}-1\}$, where $\widehat{t}_{0}$, $
\widehat{t}_{1_{i}}\in (0,1)$.  We next estimate   $\mathring{L}.$
It is required that $\phi_{a,-0.238,1,\Re}(\widehat{t}_{0})\neq0$
such that $c_{0,a}$ can be determined, up to a sign, by $|g_{a}(\widehat{t}_{0})|$. Otherwise,
$|g_{a}(\widehat{t}_{0})|$ is useless for determining  $c_{0,a}$. Without loss of generality,
assume that $c_{0,a}$ is determined and  $|g_{a}(1+\widehat{t}_{1_{i}})|\neq0.$
Next we need to determine $c_{1, a}$. Then the  determination of
$c_{1, a}$ is equivalent to the  determination of
$z:=\hbox{sgn}(g_{a}(1+\widehat{t}_{1_{1}}))$. By the    analysis in \eqref{iterative01}, $z$
is the solution to \eqref{positive} with $A^{\Re}_{1,f}(t_{1_{1}},t_{1_{2}})$ and $C^{\Re}_{1,f}(t_{1_{1}},t_{1_{2}})$
therein  replaced by
$A^{\Re}_{1,g_{a}}(\widehat{t}_{1_{1}}, \widehat{t}_{1_{j}})$ and $C^{\Re}_{1,g_{a}}(\widehat{t}_{1_{1}},\widehat{t}_{1_{j}})$, respectively,
where $j\neq1$.
Clearly, $z$ can be determined if and only if   $A^{\Re}_{1,g_{a}}(\widehat{t}_{1_{1}}, \widehat{t}_{1_{j}})\neq0$.
As an example,   we choose $\widehat{t}_{1_{1}}=0.5$ without bias.
See the graph  of $A^{\Re}_{1,g_{a}}(0.5,x)$ on $(0,1)$
in Fig. \ref{NND}.
Obviously     the number  of zeros of $A^{\Re}_{1,g_{a}}(0.5,x)$
on $(0,1)$ is much larger than $2 $.
Especially, we found that the number of zeros of  $A^{\Re}_{1,g_{50}}(0.5,x)$
is not smaller than $\textbf{256}.$ Then  we need at least $\textbf{257}$ additional  deterministic samples on $(1, 2)$ to
avoid $A^{\Re}_{1,g_{50}}(0.5,x)=0.$ Therefore  for  reconstructing  $g_{50}$, $\mathring{L}\geq \textbf{259}$
although it is determined by only   \textbf{two} coefficients.
 By Proposition \ref{main1proposition123}, however,   $g_{a}$ can be determined, with probability $1$, by just  \textbf{three} random
samples. Allowing for \eqref{jycs} we just need to check the recovery  efficiency of $f_{a,-0.2381,1}$.

In the present simulation,  by
the random samples \begin{align}\label{data123}\begin{array}{lll} \{|f_{a,b,p}(t_{0})|\}\cup\{|f_{a,b,p}(n+t_{n_{1}})|,|f_{a,b,p}(n+t_{n_{2}})|:\\
\quad  n=1, \ldots, 18\},\end{array}\end{align}  $10^3$ trials of  Approach \ref{gghh}  are conducted to recover  $f_{a,b,p}$, where   $\{t_{0}\}\cup\{t_{n_{1}},t_{n_{2}}: n=1, \ldots, 18\}\sim \textbf{U}(0, 1)$. The recovery  error is defined as
\begin{align}\label{errorspline} \begin{array}{lll} \hbox{Error}(f_{a,b,p})&:=\log_{10}(\min_{\gamma\in \{1, -1\}}||\{c_{k,a}\}\\
&\quad \quad\quad\quad-\gamma\{\widetilde{c}_{k,a}\}||_{2}/||\{c_{k,a}\}||_{2}),\end{array}\end{align}
where $\{\widetilde{c}_{k,a}\}$ is the recovery  version of $\{c_{k,a}\}.$
As in section \ref{numeri1},    the  approach   is considered   successful  if $\hbox{Error}(f_{a,b,p})\leq-1.8$, and
the  cumulative distribution function (CDF) of the
error is  defined via  \eqref{cdf}.
Clearly,    Fig. \ref{fig:balvsunbalHam167} implies  that  $f_{10,-0.2381,1}$, $f_{50,-0.2381,1}$
and $f_{10^{-6},-0.0038,0}$ can be   recovered  perfectly
 in the noiseless setting. To check the stability to noise, we also conduct  $10^3$ trials in the noisy setting. As in section \ref{numeri1}, we add  the  Gaussian noise $\varepsilon\sim \textbf{N}(0, \sigma^{2})$ to the observed  noiseless samples in \eqref{data123}.
The variance $\sigma^{2}$ is chosen via \eqref{SNRDE} with $55$ therein replaced  by $37$
such that the desired  SNR can be expressed.
As in the noiseless case, $10^3$ trials are also conducted.
The success rates  ($\hbox{CDF}(-1.8)$) are recorded in Table \ref{table123}.

Comparing Table \ref{jkk} and Table \ref{table123} we found that, for the low   SNR (e.g. $\leq 60$),
the stability to noise in the present simulation (real-valued case) is much stronger than that in
section \ref{chirp} (complex-valued case).  We next interpret this from the phase distribution perspective.

\begin{remark}
For a real-valued signal $f\in V(\varphi)$, its phase function $\theta(f(x))$ has only  two values: $0$ and $\pi$.
Since the samples   in \eqref{data123}  are perturbed,
 unavoidably  so is $\frac{C^{\Re}_{n,f}(\widehat{t}_{n_{1}},\widehat{t}_{n_{2}})}{A^{\Re}_{n,f}(\widehat{t}_{n_{1}},\widehat{t}_{n_{2}})}$
in $\textbf{step 2}$ of Approach \ref{decoder}.
If the perturbation $\epsilon$ of $\frac{C^{\Re}_{n,f}(\widehat{t}_{n_{1}},\widehat{t}_{n_{2}})}{A^{\Re}_{n,f}(\widehat{t}_{n_{1}},\widehat{t}_{n_{2}})}$ satisfies $|\epsilon|<
|\frac{C^{\Re}_{n,f}(\widehat{t}_{n_{1}},\widehat{t}_{n_{2}})}{A^{\Re}_{n,f}(\widehat{t}_{n_{1}},\widehat{t}_{n_{2}})}|$, then
$\theta(f(n+\widehat{t}_{n_{1}}))$ can be decoded \emph{exactly} through $\textbf{step 2}$.
Unlike the real-valued case, Fig. \ref{fig:balvsunbalHam1234567} implies that the phases of the complex-valued signals in section \ref{numeri1} are much more complicated. Therefore, it is  no wonder  that the stability in the present simulation is much stronger than that in
section \ref{numeri1}.
\end{remark}

On the other hand, it follows from Fig. \ref{fig:balvsunbalHam156} (b, d, f)
that the phase function  of  $f_{10^{-6},-0.0038,0}$ varies much more  slowly  than those of $f_{10,-0.2381,1}$ and $f_{50,-0.2381,1}$.
And when SNR ($\leq 50$) is low,   numerical results in Table \ref{table123} imply
the much stronger  stability for $f_{10^{-6},-0.0038,0}$.

Recall that the distribution and oscillation of the phase is the intrinsic property of a signal. Overall,
the simulation results in section \ref{numeri1} and in  the present section imply that the recovery stability to noise
is related with the property.

\section{Conclusion}\label{conclusion}
We  prove that the full spark property   is not sufficient
for the phaseless sampling  in   complex-generated shift-invariant spaces (SISs) (Theorem \ref{examm}).
We establish a   condition for decoding the phases of the samples (Theorem \ref{theorem1}).
 Based on Theorem \ref{theorem1},  we establish a reconstruction scheme  in Approach \ref{decoder}.
Based on  Approach \ref{decoder} and the  generalized Haar condition (GHC),
nonseparable and causal (NC) signals in the complex-generated SISs can be determined  with probability $1$
if  the random sampling density (SD) is not smaller than $3$ (Theorem \ref{main1}).
Approach \ref{decoder}  is modified to Approach \ref{gghh} such that it is more adaptive to
real-valued NC signals in real-generated SISs. Based on Approach \ref{gghh} and GHC,  real-valued NC signals in the real-generated SISs can be determined  with probability $1$
if  the random SD is not smaller than $2$ (Theorem \ref{main4}). Propositions \ref{main1proposition} and \ref{main1proposition123} imply  that  the highly oscillatory signals
can be   determined  locally,  with probability $1$, by a very small number of random samples.

\section{Appendix}
\subsection{Proof of Lemma \ref{theoremX}}\label{theoremXX}
Since $t_{n_{1}}, t_{n_{2}}$ and $t_{n_{3}}$ are i.i.d  random variables,
we just need to  prove $P\big(A_{n,f}(t_{n_{1}},t_{n_{2}})+B_{n,f}(t_{n_{1}},t_{n_{2}})\textbf{i}\neq0\big)=1$.

Define an event  $\widetilde{\mathfrak{E}}_{n,0}
:=\{\phi(t_{n_{1}})\bar{f}(n+t_{n_{1}})\neq0\}$ w.r.t $t_{n_{1}}$.
By \eqref{diedai}, we have
\begin{align} \notag \begin{array}{lll}\widetilde{\mathfrak{E}}_{n,0}\\
=\{\phi(t_{n_{1}})(\bar{v}_{n,f}(t_{n_{1}})+\bar{c}_{n}\bar{\phi}(t_{n_{1}}))\neq0\}\\
=\{\sum_{k\in I_{n}}\bar{c}_{k}\phi(t_{n_{1}})\bar{\phi}(n+t_{n_{1}}-k)+\bar{c}_{n}|\phi|^{2}(t_{n_{1}})\neq0\}.\end{array}\end{align}
First, it is easy to derive from Lemma \ref{nonsp} and the definition of  $\mathcal{N}_{f}$
that, for every $n\in \{1, 2, \ldots, \mathcal{N}_{f}+s-1\}$  there exists a nonzero coefficient in $\{c_{k}: k\in I_{n}\}$.
Moreover, $\Lambda_{\phi,2}$ in
Proposition \ref{Haarcond} satisfies  GHC. Then
\begin{align}\begin{array}{lll}
\mu \Big(\{t\in(0,1):\sum_{k\in I_{n}}\bar{c}_{k}\phi(t)\bar{\phi}(n+t-k)\\
\quad \quad \quad\quad \quad \quad\quad +\bar{c}_{n}|\phi|^{2}(t)=0\}\Big)\\
=0.\end{array}\end{align}
Therefore $P(\widetilde{\mathfrak{E}}_{n,0})=1$. Consequently,   $P(\mathfrak{E}_{n,0})=1$ where $\mathfrak{E}_{n,0}=\big\{\frac{|f(n+t_{n_{1}})|}{|\phi|^{2}(t_{n_{1}})}\neq0\big\}$.
Define an auxiliary (random) function  w.r.t  $t_{n_{1}}$ and $t_{n_{2}}$ by
\begin{align} \label{normal}\begin{array}{lll}  a_{n,f}(t_{n_{1}},t_{n_{2}})+b_{n,f}(t_{n_{1}},t_{n_{2}})\textbf{i}\\
:=\bar{\phi}(t_{n_{1}})\phi(t_{n_{2}})\bar{v}_{n,f}(t_{n_{2}})-\bar{v}_{n,f}(t_{n_{1}})|\phi|^{2}(t_{n_{2}}).\end{array}\end{align}
Direct observation on \eqref{Teren} leads to that
 \begin{align} \label{normal89}\begin{array}{lll} A_{n,f}(t_{n_{1}},t_{n_{2}})+B_{n,f}(t_{n_{1}},t_{n_{2}})\textbf{i}\\
 =\frac{|f|(n+t_{n_{1}})}{|\phi|^{2}(t_{n_{1}})}\big(a_{n,f}(t_{n_{1}},t_{n_{2}})+b_{n,f}(t_{n_{1}},t_{n_{2}})\textbf{i}\big).\end{array}\end{align}
As previously   for every $n\in \{1, 2, \ldots, \mathcal{N}_{f}\}$,  there exists a nonzero coefficient in $\{c_{k}: k\in I_{n}\}$.
Then by \eqref{vn} we have $\bar{v}_{n,f}\not\equiv0.$ Now it follows from
$\Lambda_{\phi,2}$ in
Proposition \ref{Haarcond} satisfying GHC that $\phi\bar{v}_{n,f}$ and $|\phi|^{2}$
are linearly independent, which together with $P(\mathfrak{E}_{n,0})=1$
leads to    $a_{n,f}(\cdot,\cdot)+b_{n,f}(\cdot,\cdot)\textbf{i}\not\equiv0.$
Then
\begin{align} \notag \begin{array}{lll} 1\geq P\big(a_{n,f}(t_{n_{1}},t_{n_{2}})+b_{n,f}(t_{n_{1}},t_{n_{2}})\textbf{i}\neq0\big)\\
\ \   \geq  P\big(a_{n,f}(t_{n_{1}},t_{n_{2}})+b_{n,f}(t_{n_{1}},t_{n_{2}})\textbf{i}\neq0|\mathfrak{E}_{n,0}\big)P(\mathfrak{E}_{n,0})\\
\ \ = P\big(a_{n,f}(t_{n_{1}},t_{n_{2}})+b_{n,f}(t_{n_{1}},t_{n_{2}})\textbf{i}\neq0|\mathfrak{E}_{n,0}\big)\\
\ \ =1,\end{array}\end{align}
where $\Lambda_{\phi,2}$   satisfying GHC is used again in the last identity.
The proof is concluded.

\subsection{Proof of Lemma \ref{theo2.5}}\label{XYZ}
If $0<\frac{|f|(n+t_{n_{1}})}{|\phi|^{2}(t_{n_{1}})}<\infty$,  then it follows from \eqref{normal89}  that  $\theta[A_{n,f}(t_{n_{1}}, t_{n_{2}})+B_{n,f}(t_{n_{1}}, t_{n_{2}})\textbf{i}]$
$=\theta[a_{n,f}(t_{n_{1}}, t_{n_{2}})+b_{n,f}(t_{n_{1}}, t_{n_{2}})\textbf{i}]$, where $a_{n,f}(t_{n_{1}}, t_{n_{2}})+b_{n,f}(t_{n_{1}}, t_{n_{2}})\textbf{i}$ is defined in \eqref{normal}.
By direct calculation, for $y\in (0, 1)$  we have
\begin{align}\notag \begin{array}{lll} \Re(a_{n,f}(t_{n_{1}},y)+\textbf{i}b_{n,f}(t_{n_{1}},y))\\
=a_{n,f}(t_{n_{1}},y)\\
=u_{t_{n_{1}},f}(\phi^{2}_{\Re}(y)+\phi^{2}_{\Im}(y))\\
\ +\sum_{k\in I_{n}}[\widetilde{c}_{t_{n_{1}},k,\Re}\big(\phi_{\Re}(y)\phi_{\Re}(y+n-k)\\
\quad\quad\quad\quad +\phi_{\Im}(y)\phi_{\Im}(y+n-k)\big)]\\
\ -\sum_{k\in I_{n}}[\widetilde{c}_{t_{n_{1}},k,\Im}\big(\phi_{\Im}(y)\phi_{\Re}(y+n-k)\\
\quad\quad\quad\quad-\phi_{\Re}(y)\phi_{\Im}(y+n-k)\big)],\end{array}\end{align}
and
\begin{align}\notag \begin{array}{lll} \Im(a_{n,f}(t_{n_{1}},y)+\textbf{i}b_{n,f}(t_{n_{1}},y))\\
=b_{n,f}(t_{n_{1}},y))\\
=v_{t_{n_{1}},f}(\phi^{2}_{\Re}(y)+\phi^{2}_{\Im}(y))\\
\ +\sum_{k\in I_{n}}[\widetilde{c}_{t_{n_{1}},k,\Im}\big(\phi_{\Re}(y)\phi_{\Re}(y+n-k)\\
\quad \quad \quad \quad +\phi_{\Im}(y)\phi_{\Im}(y+n-k)\big)]\\
\ +\sum_{k\in I_{n}}[\widetilde{c}_{t_{n_{1}},k,\Re}\big(\phi_{\Im}(x)\phi_{\Re}(y+n-k)\\
\quad \quad \quad \quad -\phi_{\Re}(y)\phi_{\Im}(y+n-k)\big)],\end{array}\end{align}
where $\bar{v}_{n,f}(t_{n_{1}}):=u_{t_{n_{1}},f}+\textbf{i}v_{t_{n_{1}},f} $ and
 \begin{align}\label{gggg} \widetilde{c}_{t_{n_{1}}, k}:=\bar{\phi}(t_{n_{1}})c_{k}=\widetilde{c}_{t_{n_{1}},k,\Re}+\textbf{i}\widetilde{c}_{t_{n_{1}},k,\Im}.\end{align}

As mention in section \ref{theoremXX},  there exists at least one nonzero coefficient in $\{c_{k}: k\in I_{n}\}$ for every $n\in \{1, 2, \ldots, \mathcal{N}_{f}+s-1\}$. For
any fixed  $n\in \{1, \ldots, \mathcal{N}_{f}\}$,
using    $\Lambda_{\phi,1}$ in Proposition  \ref{Haarcond} satisfying  GHC, we have
$P(\bar{\phi}(t_{n_{1}})\neq0)=1$, which together with \eqref{gggg} leads to that    with probability  $1$,  there exists at least one nonzero coefficient in $\{\widetilde{c}_{t_{n_{1}}, k}: k\in I_{n}\}$.
Then
\begin{align}\notag \begin{array}{lll}
P\big(\Re(a_{n,f}(t_{n_{1}},t_{n_{2}})+\textbf{i}b_{n,f}(t_{n_{1}},t_{n_{2}}))\neq0\big)\\
\geq P\big(\Re(a_{n,f}(t_{n_{1}},t_{n_{2}})+\textbf{i}b_{n,f}(t_{n_{1}},t_{n_{2}}))\neq0|\mathfrak{E}_{n,0}\big)\\
\ \ \times P(\mathfrak{E}_{n,0})\\
=P\big(\Re(a_{n,f}(t_{n_{1}},t_{n_{2}})+\textbf{i}b_{n,f}(t_{n_{1}},t_{n_{2}}))\neq0|\mathfrak{E}_{n,0}\big)\\
=1,
\end{array}\end{align}
where $P(\mathfrak{E}_{n,0})=1$,   derived from section \ref{theoremXX}, is used in the first identity, and the second identity
is derived from   GHC \eqref{quantity11}. Therefore, $P\big(\Re(a_{n,f}(t_{n_{1}},t_{n_{2}})+\textbf{i}b_{n,f}(t_{n_{1}},t_{n_{2}}))\neq0\big)=1.$
Similarly, we can prove that $P\big(\Im(a_{n,f}(t_{n_{1}},t_{n_{2}})+\textbf{i}b_{n,f}(t_{n_{1}},t_{n_{2}}))\neq0\big)=1.$
Then $P\big(\theta[a_{n,f}(t_{n_{1}}, t_{n_{2}})+b_{n,f}(t_{n_{1}}, t_{n_{2}})\textbf{i}]=\frac{j\pi}{2}\big)=0,$
where $j=0, 1, 2,3.$ Applying the above result to
$\widetilde{f}:=e^{\textbf{i}(\frac{\pi}{2}-\alpha)}f\in V_{\hbox{ca}}(\phi)$, the proof is concluded.

\subsection{Proof of Lemma \ref{auxill0}}\label{XYZZZ}
Define  three  random  events
\begin{align} \label{danding09} \begin{array}{lll}  \mathfrak{E}_{1}:=\Big\{(A_{n,f}(t_{n_{1}}, t_{n_{2}})+\textbf{i}B_{n,f}(t_{n_{1}}, t_{n_{2}}))\\
\quad \quad \quad  \ \times (A_{n,f}(t_{n_{1}}, t_{n_{3}})-\textbf{i}B_{n,f}(t_{n_{1}}, t_{n_{3}}))\\
\quad \quad  \quad\neq(A_{n,f}(t_{n_{1}}, t_{n_{2}})-\textbf{i}B_{n,f}(t_{n_{1}}, t_{n_{2}}))\\
\quad \quad \quad \ \times (A_{n,f}(t_{n_{1}}, t_{n_{3}})+\textbf{i}B_{n,f}(t_{n_{1}}, t_{n_{3}}))\Big\}, \end{array} \end{align}
and
\begin{align}\begin{array}{lll} \label{fendou}  \mathfrak{E}_{2}:=\{A_{n,f}(t_{n_{1}}, t_{n_{2}})+\textbf{i}B_{n,f}(t_{n_{1}}, t_{n_{2}})\neq0\}, \\ \mathfrak{E}_{3}:=\{A_{n,f}(t_{n_{1}}, t_{n_{3}})+\textbf{i}B_{n,f}(t_{n_{1}}, t_{n_{3}})\neq0\}.\end{array} \end{align}
Next we  prove that  $P(\mathfrak{E}_{1})=1$.
By Lemma  \ref{theoremX}, $P(\mathfrak{E}_{2})=P(\mathfrak{E}_{3})=1 $.
 Direct computation gives that
 \begin{align}\label{prob1}\begin{array}{lll}
 1&\geq P(\mathfrak{E}_{1})\\
 &\geq P(\mathfrak{E}_{1}\cap\mathfrak{E}_{2})\\
 &=P(\mathfrak{E}_{1}|\mathfrak{E}_{2})P(\mathfrak{E}_{2})\\
 &=P(\mathfrak{E}_{1}|\mathfrak{E}_{2}).
 \end{array}
 \end{align}
 By \eqref{danding09} and \eqref{fendou}, we have
  \begin{align}\notag \begin{array}{lll} \mathfrak{E}_{1}|\mathfrak{E}_{2}\\
  =\Big\{A_{n,f}(t_{n_{1}}, t_{n_{3}})-\textbf{i}B_{n,f}(t_{n_{1}}, t_{n_{3}})\\
  \  -b(t_{n_{1}}, t_{n_{2}})(A_{n,f}(t_{n_{1}}, t_{n_{3}})+\textbf{i}B_{n,f}(t_{n_{1}}, t_{n_{3}}))\neq0|\mathfrak{E}_{2}\Big\},\end{array}
 \end{align}
where \begin{align}\notag \begin{array}{lll} b(t_{n_{1}}, t_{n_{2}})=\frac{A_{n,f}(t_{n_{1}}, t_{n_{2}})-\textbf{i}B_{n,f}(t_{n_{1}}, t_{n_{2}})}{A_{n,f}(t_{n_{1}}, t_{n_{2}})+\textbf{i}B_{n,f}(t_{n_{1}}, t_{n_{2}})}.\end{array}
 \end{align}
Applying  Lemma \ref{theo2.5} to $A_{n,f}(t_{n_{1}}, t_{n_{3}})+\textbf{i}B_{n,f}(t_{n_{1}}, t_{n_{3}})$,  it is easy to prove that $P(\mathfrak{E}_{1}|\mathfrak{E}_{2})=1 $ which together with \eqref{prob1} leads to   $P(\mathfrak{E}_{1})=1.$
Now the rest of  proof can be easily    concluded.

\begin{IEEEbiography}{Youfa Li}
Biography text here.
\end{IEEEbiography}

\begin{IEEEbiography}{Wenchang Sun}
Biography text here.
\end{IEEEbiography}

\tabcolsep 30pt

\vspace*{7pt}

\end{document}